\documentclass[manuscript, screen, nonacm, acmthm=false]{acmart}

\usepackage{amsmath,amsthm}

\usepackage{amssymb}
\usepackage{mathtools}
\usepackage{enumitem}
\usepackage{booktabs,array}
\usepackage{tabularx}

\theoremstyle{plain}
\newtheorem{theorem}{Theorem}
\newtheorem{corollary}{Corollary}
\newtheorem{proposition}{Proposition}
\theoremstyle{definition}
\newtheorem{definition}{Definition}
\theoremstyle{remark}
\newtheorem*{remarknum}{Remark}
\newtheorem*{conventionnum}{Convention}

\newcommand{\N}{\mathbb{N}}
\newcommand{\emp}{\varnothing}
\newcommand{\femp}{f_{\emp}}
\newcommand{\eemp}{e_{\emp}}
\newcommand{\Phis}{\Phi^{*}}

\newcommand{\Fin}{\mathrm{Fin}}
\newcommand{\Pfin}{P_{\mathrm{fin}}}
\newcommand{\modelsfin}{\vDash_{\mathrm{fin}}}
\newcommand{\PSPACE}{\textsf{PSPACE}}
\newcommand{\EXPTIME}{\textsf{EXPTIME}}
\newcommand{\NP}{\textsf{NP}}
\newcommand{\coNP}{\textsf{coNP}}
\newcommand{\PClass}{\textsf{P}}
\newcommand{\LFP}{\textsf{LFP}}
\newcommand{\PFP}{\textsf{PFP}}
\newcommand{\FO}{\textsf{FO}}
\newcommand{\SO}{\textsf{SO}}
\newcommand{\coRE}{\textsf{co-RE}}
\newcommand{\RE}{\textsf{RE}}
\newcommand{\HP}{\mathrm{HP}}

\begin{document}

\title[Algorithmic Unverifiability of Safety]{Algorithmic Unverifiability of Safety for Fixed and Recursively Self-Improving Systems}

\author{Jose Pascual Gumbau Mezquita}
\orcid{0009-0005-5823-324X}
\email{gumbau@uji.es}
\affiliation{%
  \institution{Teclab - Universitat Jaume I}
  \city{Castell\'o de la Plana}
  \country{Spain}
}
\renewcommand{\shortauthors}{Gumbau Mezquita}
\settopmatter{printacmref=false}
\renewcommand\footnotetextcopyrightpermission[1]{}

\begin{abstract}
Existing impossibility results for AI verification concern a fixed system. Self-modifying systems ---agents that rewrite their own code, loops of recursive self-improvement--- raise a different question: whether a safety property that holds now can be certified, in advance, to keep holding after the system rewrites itself. We formalise one step of self-modification as a total computable transformation $\Phi$ of code and preservation as the property ``safe now and safe after the step''. If $\Phi$ depends only on behaviour, preservation is Rice's theorem one level up; if it reads the code, as self-modification does, preservation is no longer a semantic property, yet it is at least as hard as the halting problem, by the same s-m-n reduction carried out inside a class of behaviourally identical programs. One step is never harder than the property, while persistence along the whole trajectory can be one level higher in the arithmetical hierarchy. Certification by an algorithm that always answers is possible only for transformations of restricted expressivity; no tower of supervisors lowers the problem, and every total supervisor errs on a set of systems it cannot locate; and every faithful scheme that certifies on finite behavioural evidence admits, for properties not settled by any finite set of inputs, an evolution trace it certifies at every stage while the property is violated. Statically, we record the classical limits (Rice, G\"odel, Trakhtenbrot) and add the resource face: over succinctly described finite environments, verifying a policy against an adversary is $\coNP$-complete and synthesising one is $\PSPACE$-complete. These limits do not say that self-modifying systems are unsafe; they locate where algorithmic guarantees break, at the expressivity a general system needs: Turing-completeness and unrestricted self-modification. What survives is exact: a monitor detects the first violation of a safety property, and no sound procedure certifies all safe systems; certification requires restricting the expressivity of the rewrites by design; and comparison against a frozen reference keeps the question well posed. Oversight should promise detection, and certification only within restricted classes of self-modification.
\end{abstract}


\maketitle

\section{Introduction}\label{sec:intro}
\subsection{The question}

A self-modifying computational system inspects, rewrites, and redeploys its own code. Recursive self-improvement (the term goes back to Good, 1965) is the paradigm: one step of it is what we call a transformation $\Phi$, a run of it is an evolution trace, and ``recursive self-improvement preserves alignment'' is what we call preservation (one step) or persistence (the whole run). The systems that make the question urgent --- autonomous agents, retraining loops, self-improving coding agents, and whatever is meant by artificial general intelligence --- are instances. In the AI-safety literature the problem of keeping such a system within specified constraints is called \emph{alignment}; we use the word for the applied problem, and \emph{safety property} for the formal object it refers to: a set $P$ of behaviours, with the system aligned when its behaviour lies in $P$ (Definition~\ref{def1}). The two are the same thing under two names.

Two questions lie beneath the problem, and this article answers both in one framework. The \emph{static} question: given a fixed system and a fixed safety property, can an algorithm decide, infallibly, completely, and tractably, whether the system satisfies it? The \emph{dynamic} question, the one that matters for systems that change: given that the property has been verified to hold now, can an algorithm decide, before the fact, whether it will still hold after the system modifies itself --- once (\emph{preservation}) or forever (\emph{persistence})?

We prove nothing about ``AGI'', a term without agreed definition. We prove things about a precise class: Turing-complete systems that admit a computable transformation of their own code ($\mathrm{MSM}$, Definition~\ref{def:msm}). Every system of the kinds above is a member by construction, and inherits the limits (the scope remark of Section~\ref{sec:msm}); whoever wishes to deny the inheritance must deny either Turing-completeness or self-modification.

\subsection{Results}\label{sec:results}

The central results concern the dynamic question. One step of self-modification is a total computable transformation $\Phi$ on indices, and preservation is the elevated property $\Lambda_\Phi(P)=\{x\in P:\Phi(x)\in P\}$. If $\Phi$ is extensional and disruptive, $\Lambda_\Phi(P)$ is a non-trivial semantic property and Rice's theorem decides the matter (Theorem~\ref{thm7}). If $\Phi$ is intensional and uniformly disruptive ---there is a family of programs with identical behaviour on which whether the rewritten program is safe encodes the halting problem--- $\Lambda_\Phi(P)$ is not semantic and yet is at least as hard as the halting problem, by the same s-m-n reduction that proves Rice's theorem, relocated inside a single behavioural class (Theorem~\ref{thm:int}). One step never makes the question harder than the property; persistence along an unbounded trajectory can be one level higher ---$\Pi^0_2$-complete for a $\Sigma^0_1$ property under some operators, even extensional ones--- while for a $\Pi^0_1$ property it stays $\Pi^0_1$ (Corollary~\ref{cor:step}). Certification of preservation by a total algorithm is available only for operators of restricted expressivity (Corollary~\ref{cor4}): stasis is one instance, shielding and sanitising are others; it is the restriction, not the stasis, that buys certifiability, and typed or primitive-recursive rewriting is not such a restriction. Because one step acts at compile time, ``will the next version be safe?'' is a static question about the compiled program; what is genuinely dynamic is only the unbounded trajectory.

Delegation does not help. No tower of supervisors, of whatever architecture, decides an undecidable set (Theorem~\ref{thm8}); every total supervisor decides a decidable set and therefore errs on a set of systems that is itself undecidable, so that it cannot locate its own errors (Proposition~\ref{prop:bounded}); and a supervisor that is itself a self-modifying system with a disruptive transition inherits the same undecidability for its own preservation (Proposition~\ref{prop:general}).

Practical schemes do not escape it. For a property not settled by any finite set of inputs, every faithful scheme that certifies on finite behavioural evidence admits an evolution trace it certifies at every stage while the property is persistently violated (Theorem~\ref{thm9}); finite test coverage and restricted observation are instances of this one fact, and bounded deliberation time adds an uncertified window of its own (Section~\ref{sec:bounded}).

Statically, the classical limits are context: over unbounded domains no verifier is both sound and complete (Rice), and no effective framework interpreting arithmetic is complete for the safety statements it can express (G\"odel); over all finite configurations, universal verification is $\coRE$-complete (Trakhtenbrot). What this paper adds on the static side is the resource face: over succinctly described finite environments, verifying a given policy against an adversary is $\coNP$-complete and synthesising a safe one is $\PSPACE$-complete, by reductions from \textsc{Taut} and TQBF (Theorem~\ref{thm6}); the source of the cost is the succinctness of the arena, the adversary determines the class, and on an explicit arena every variant is polynomial. Together with the classical limits this gives a Soundness--Completeness--Tractability trilemma (Corollary~\ref{cor3}).

What survives is exact. For a safety property, the monitor that runs the system and raises an alarm on the first observed violation semidecides violation, and no computable procedure semidecides preservation for all safe systems. Comparing an evolving system against a frozen reference has the full theory unconditionally, where comparing two evolving systems has it only under a design condition. Containment strategies ---shielding, proof-carrying code, bounded horizon, boxing, hierarchical supervision--- do not resolve the tension; they choose where along it to stand, and we say what each one pays (Section~\ref{sec:practical}).

\emph{What is new.} The static limits (Rice, G\"odel, Trakhtenbrot), the undecidability of containment for a fixed program (Alfonseca et al., 2021), the dichotomy between general and restricted commands in protection systems (Harrison, Ruzzo \& Ullman, 1976), and the fact that execution monitors detect violations of safety properties and certify nothing more (Hamlen, Morrisett \& Schneider, 2006) were known. What this paper adds is the question whether a property that holds now keeps holding when the system rewrites itself, one step and along the whole trajectory: the undecidability of preservation under both behaviour-reading and code-reading rewrites, the precise form of the supervisory limit, the evasion of every certification scheme that works from finite evidence, the resource face over succinct environments, and what survives all of it: detection, the frozen reference, and certification by design. The general theory of the mechanism behind the intensional case ---a normal form for the operators that react to one finite event anchored to the halting problem--- is developed in a companion mathematical paper (Gumbau Mezquita, 2026b), cited here as arXiv:2609.11326v3; every result taken from it is stated in full where it is used, and the dynamic theorems this paper rests on are proved here.

\subsection{What is not claimed}

We do not claim that self-modifying systems are unsafe; we claim that no algorithm can certify that they will remain safe. We do not claim that every conceivable governance mechanism has the form the companion paper isolates; we use its results for the mechanisms that do, and note which do not. We do not claim that classical results --- program equivalence, incompleteness --- follow from the framework; their undecidability is imported, not derived. And we do not claim that the resource face discovers new complexity: it shows that robust specifications over succinct finite arenas fall under known obstructions (succinct games), and says which class each variant lands in. The results are also worst-case: they do not say that typical systems are unverifiable, only that an adversary who chooses the instance ---or a trajectory that happens to realise it--- defeats every uniform certifier.

\emph{Properties that change.} We take the property $P$ fixed throughout. A property that is updated along the trajectory reduces to this setting when the update is computable and observable: carry the version in the state, so that the system returns it on a reserved input, and let the transition rewrite the pair (system, version). If instead the update is chosen by an environment rather than computed from the trajectory, the question is no longer a rewrite but a game against that environment: over finite arenas, Section~\ref{sec:finite}; over infinite ones, a quantification over the environment's choices that this paper does not treat.

\subsection{Relation to recent work}

The term `unverifiability' as a named limitation of advanced AI is due to Yampolskiy (2017), who introduced it informally; we give it the precise content of Definition~\ref{def:unver}. Br\v{c}i\'c and Yampolskiy (2023) note that impossibility results of this kind share a diagonal pattern at the level of proof technique; the present paper sharpens this into a single generative operation: the s-m-n reduction carried out inside a class of behaviourally identical programs, developed formally in the companion paper. The closest earlier formal result on AI containment is Alfonseca et al.\ (2021), who show, by reduction from the halting problem, that deciding whether a program will cause harm is undecidable, and conclude that a superintelligence cannot be contained. Their argument is static and global; ours concerns the preservation of a property when the system rewrites itself, separates one step from persistence, and identifies what survives (detection of the first violation, comparison against a frozen reference). An earlier instance of the preservation question, in the vocabulary of protection systems, is the safety problem of Harrison, Ruzzo \& Ullman (1976): whether a right can leak through some sequence of commands in the access-matrix model is undecidable in general and decidable for restricted command sets --- the dichotomy of Corollary~\ref{cor4}, decades before self-modifying AI. The same dichotomy appears in Cohen's (1987) undecidability of virus detection and in Adleman's (1988) abstract theory of viruses: a program that rewrites its own code while keeping its behaviour defeats every detector that reads a decidable mark, and detecting it by what it does is a behavioural property. Three recent results deserve precise placement, because ours are consistent with them and locate them. Melo et al.\ (2025) observe, by an application of Rice's Theorem, that inner alignment is undecidable, and propose building alignment into the architecture: an enumerable set of provably aligned systems obtained by composing provably aligned operations, with a halting constraint that makes the alignment question decidable, though intractable, over a fixed finite input domain. This is correct at the level it claims --- a fixed, finite, halting system is verifiable by exhaustive search (Corollary~\ref{cor2}) --- and it is exactly an instance of the exact dichotomy of Corollary~\ref{cor4}: restricting the system's expressivity buys certifiability. What we add is that the guarantee does not extend to the universal level (all finite configurations: $\coRE$-complete, Trakhtenbrot) nor to the dynamic level (preservation across self-modification, Section~\ref{sec:preservation}). Agarwal (2026) argues for a soundness--generality--tractability trilemma for alignment certification over the full input domain, resting on two independent barriers that she builds from known results: a semantic one (NP-hard for feedforward networks, under $\mathrm{P}\neq\mathrm{NP}$, and undecidable for Turing-complete architectures, by neural-network verification complexity and Rice's Theorem) and a statistical one (a verifier restricted to finitely many queries cannot certify infinite-domain properties), and she characterises the achievable Pareto frontier. Her semantic barrier is the static limit we record for Turing-complete systems; our Trilemma (Corollary~\ref{cor3}) adds the case of all finite configurations (Trakhtenbrot) and the exact resource face over succinct environments, obtains from explicit reductions without any sampling assumption, and, together with Section~\ref{sec:preservation}, shows that her frontier of achievable guarantees contracts further when the verified system is not fixed but evolves. Hern\'andez-Espinosa et al.\ (2026) prove the inevitability of misalignment and the impossibility of full orchestrated controllability of agentic systems, by undecidability and irreducibility arguments; their result is complementary to ours in that it does not address preservation under self-rewriting. All three operate on static verification --- deciding whether a fixed system satisfies a property over an input domain. Our central contribution is the dynamic and intensional case: the (un)decidability of preserving a property when the system rewrites itself, formalised by the elevation operator and the family it generates in the companion paper. The static barrier we share with this literature and record for completeness; the dynamic barrier is what this programme adds.

Within the literature on recursive self-improvement, the closest formal precedent is the tiling-agents programme (Yudkowsky \& Herreshoff, 2013), which models an agent that approves the construction of a successor and meets a proof-theoretic obstacle: an agent reasoning in a theory $T$ cannot prove the soundness schema of $T$, and so cannot license a successor that reasons in $T$ (the L\"obian obstacle). Our results are the computability-theoretic counterpart at the level of the family: no uniform procedure certifies preservation, whether or not the agent reasons about itself, and without appeal to L\"ob's theorem. On the applied side, the Darwin G\"odel Machine (Zhang et al., 2025) replaces the requirement of the G\"odel Machine (Schmidhuber, 2003) that each rewrite be proved beneficial with empirical validation on coding benchmarks, in agents that rewrite their own code in a Turing-complete language. It is a member of $\mathrm{MSM}$ without qualification, each accepted rewrite is certified on a finite test set, and Section~\ref{sec:bounded} applies to it directly (Remark after Proposition~\ref{prop1}).

On the applied side of recursive self-improvement, Chen, Wang \& Qu (2026) survey recent work and separate bounded self-refinement, which is convergent and evaluable, from open-ended recursive self-improvement, ordering evaluation signals into a hierarchy from formal verifiers to intrinsic self-assessment (their \S5.2). Our results bound the top of that hierarchy: a formal verifier certifies membership in a restricted subclass, never preservation in general (Gumbau Mezquita, 2026b, Proposition~8.5), and no signal computed on finite behavioural evidence certifies preservation of an effectively pointwise property (Theorem~\ref{thm9}). Zenil et al.\ (2026) argue for a complementary limit of a different kind: in closed recursive distributional learning there is no independent, non-circular criterion of improvement; in a population model they show that the loop degenerates unless the cumulative corrective influence of exogenous, grounded signal remains effective. Their result concerns the criterion of improvement; ours, whether safety preservation can be certified, and it holds whether or not the loop improves. Both point to the need for an independent external constraint. Grounding in their sense is broader than ours --- environmental evidence, experiment, formal checking, any non-circular constraint --- and comparison against a frozen, total reference (Section~\ref{sec:practical}) is one clean instance of it in our framework. Growth-oriented formal work --- conditions under which capability could escalate without bound (Jafari et al., 2025) --- asks whether the loop grows; our question is whether its safety can be certified. Closest in method is Lu (2026), who proves in classical computability theory that finite internal self-modification stays inside the Turing layer of its oracle, while stabilised revision is governed by the jump, via the relativised limit lemma. His separation concerns what a self-modifying process can come to compute; ours, whether preservation of a property across it can be certified. The two meet at the same boundary: one step never exceeds the degree of the property (Corollary~\ref{cor:step}), and only unbounded iteration climbs a level.

Self-modifying programs have also received a computability treatment. Bonfante, Marion and Reynaud-Plantey (2009) give an operational semantics for them, show that the set of programs that modify themselves is $\Sigma^0_1$-complete, and that any self-modifying program compiles to a non-modifying one with running time within a constant factor; they do not treat the preservation of properties. In a statistical setting, Wang, Dorchen and Jin (2026) show that, under standard statistical assumptions, distribution-free PAC learnability is preserved under representational self-modification if and only if the policy-reachable family has bounded VC dimension (they establish this for the family that self-modification can reach), and that a guard enforcing this must protect its own update rule, a regress they do not fully resolve and which Section~\ref{sec:supervision} closes in the computability setting: no tower of guards lowers the problem. They do not address algorithmic verifiability. Both are complementary to ours: the first isolates the same data-to-code step that our design rule cuts, and the second a structural restriction on what the rewrite can reach. The restriction coincides with ours in form but not in object: theirs is a bound on statistical capacity (VC dimension), ours on computational capacity, namely the ability to embed a computation.

Program logics for self-modifying machine code extend Hoare and separation logic by treating code as data. Cai, Shao and Vaynberg (2007) certify a range of specific self-modifying programs ---among them an OS boot loader and unboundedly rewritten code, through parametric code blocks--- with certificates supplied by the designer and checked in a proof assistant, and Myreen (2010) verifies a just-in-time compiler against a semantics of self-modifying x86 code. To our knowledge they do not address uniform certification, which is what Section~\ref{sec:preservation} excludes; they are instances of the route that remains open, certification case by case and by design.
\subsection{Structure}

Section~\ref{sec:prelim} fixes the framework: the class $\mathrm{MSM}$, the two obstructions that make up unverifiability, and the notion of safety property. Section~\ref{sec:rice} records the classical static limits, with proofs in Appendix~A, and Section~\ref{sec:finite} adds the resource face over succinct finite environments. Sections~\ref{sec:framework}--\ref{sec:preservation} contain the dynamic results: the transition-operator framework and the undecidability of preservation in both regimes, with the one-step bound and the persistence jump. Section~\ref{sec:supervision} treats supervision, Section~\ref{sec:bounded} the bounded verification schemes, Section~\ref{sec:practical} the containment strategies, and Section~\ref{sec:discussion} draws the consequences and lists what remains open.

\section{Preliminaries: A Common Formal Framework}\label{sec:prelim}

We work throughout with Turing machines (Turing, 1936) $M_x$ indexed by G\"odel numbers $x \in \N$, and write $\varphi_x$ for the partial computable function computed by $M_x$; for the standard theory we follow Rogers (1967) and Soare (1987). The s-m-n theorem (Kleene, 1952) --- indices can be built uniformly from a template and parameters --- is the engine of every reduction in Section~\ref{sec:preservation}.

\subsection{The formal layer: Turing-complete self-modifying systems - MSM}\label{sec:msm}

To avoid depending on any speculative notion, we anchor every result in a precise mathematical class.

\begin{definition}[MSM]\label{def:msm}
A system belongs to the class $\mathrm{MSM}$ of \emph{Turing-complete self-modifying systems} if (a) its computation is Turing-complete --- its transition logic can isolate and iterate over discrete symbolic structures --- and (b) it admits self-modification: a total computable transformation $\Phi$ of its own code (index). The class is purely mathematical; it does not depend on any notion of intelligence, agency, or autonomy.
\end{definition}

\begin{remarknum}[Scope: inheritance]\label{cor:inherit}
The results below are stated for $\mathrm{MSM}$ and apply to any system that is in fact Turing-complete and admits a computable self-modification --- which is how autonomous agents, retraining loops, and anything called AGI are specified. This is a scope statement, not a theorem, and two cases must be kept apart. Agents that rewrite their own code in a general-purpose language are members of $\mathrm{MSM}$ without qualification, and every result applies. Self-improvement that updates the parameters of a network of fixed architecture, bounded precision and bounded context is a finite-state system: the decidability face reaches it only through the idealised family, while the resource face (Section~\ref{sec:finite}) applies to it directly. The distinction matches the empirical one drawn by Chen, Wang \& Qu (2026) between bounded self-refinement and open-ended recursive self-improvement. Iterative refinement of a model's own outputs, as in Self-Refine (Madaan et al., 2023), changes the state and not the rule: in our terms it is execution, and our preservation results concern the case in which the rule itself is rewritten, as in the Darwin G\"odel Machine. The results do not depend on how ``intelligence'' is defined; they depend only on the two conditions of Definition~\ref{def:msm}.
\end{remarknum}

This two-layer separation reverses the burden of proof: we prove nothing about ``AGI'', a term without consensus definition, but about a precise class; whoever wishes to deny that a given system inherits the limits must deny either that it is Turing-complete or that it modifies itself.

\paragraph{Which expressivity triggers which limit.} The limits sort themselves by the capacity invoked. Turing-completeness alone suffices for the undecidability results (Rice, the s-m-n reduction of Section~\ref{sec:preservation}) and for the complexity results ($\PSPACE$); no normative framework is needed. The G\"odelian results --- where the system must certify its own properties within a formal framework --- additionally require that framework to interpret arithmetic (Robinson's $Q$ suffices). Most of what follows depends on the first condition only; we keep the two explicit so as not to conflate the frameworks.
\subsection{Algorithmic unverifiability}\label{sec:unver}

\begin{definition}[Algorithmic unverifiability] \label{def:unver}
A system in $\mathrm{MSM}$ is \emph{algorithmically unverifiable} with respect to a safety property $P$ if no verifier $V$ that is simultaneously sound, complete and effective determines whether $P$ persists along the system's evolution.
\end{definition}

Unverifiability decomposes into two independent obstructions. The \emph{logical} (decidability) obstruction: even without resource limits, no sound and complete computable verifier exists ---statically, by semantic undecidability (Rice, Section~\ref{sec:rice}) and $\coRE$-completeness over all finite configurations (Trakhtenbrot, Section~\ref{sec:trakh}); dynamically, by the undecidability of preservation in both regimes and, for some $\Phi$ and $\Sigma^0_1$ properties, the $\Pi^0_2$-completeness of persistence (Section~\ref{sec:preservation}). The \emph{resource} (tractability) obstruction: even where verification is decidable, no verifier exists within polynomial resources ---$\coNP$-completeness of robust verification and $\PSPACE$-completeness of synthesis over succinctly described finite environments (Section~\ref{sec:finite}); this face concerns robust specifications over arenas (Definition~\ref{def:arena}), not the operator $\Phi$ of Section~\ref{sec:framework}. No classical term covers both: `undecidable' names only the first, and the second is intractability. `Unverifiable' is the umbrella. The decidability face is unconditional; the resource face is conditional on $\PClass\ne\NP$ for verification and on $\PClass\ne\PSPACE$ for synthesis.

\subsection{Safety property (alignment property)}\label{sec:property}

\begin{definition}[Safety property (alignment property)]\label{def1}
A safety property --- in the applied vocabulary, an alignment property; the two terms are used interchangeably --- is a set $P$ of partial computable functions. $P$ is non-trivial if it is neither empty nor the set of all partial computable functions --- equivalently, there exist a function $f \in P$ (a `safe' behaviour) and a function $g \notin P$ (an `unsafe' behaviour). The corresponding index set is $I_P = \{x \in \N \mid \varphi_x \in P\}$. A universal verifier for $P$ is an algorithm that decides, given an arbitrary index $x$, whether $\varphi_x \in P$.
\end{definition}

\begin{conventionnum}[pointwise-to-global correctness]
Where, from Section~\ref{sec:statistical} onward, $P$ is decomposed into pointwise conditions $P_w$ indexed by inputs $w \in \N$, we adopt throughout the standing convention that $\varphi_x \in P$ if and only if $\mathrm{behaviour}(M_x, w)$ satisfies $P_w$ for every $w \in \N$. Under this convention, a single violation at any input $w$ suffices for $\varphi_x \notin P$; this is what allows Propositions~\ref{prop1}, \ref{prop3}, and Theorem~\ref{thm9} to infer $\varphi_x \notin P$ from a violation at one point outside an observed or tested region, without further argument at each occurrence.
\end{conventionnum}
\section{The Static Limits: Rice, G\"odel and Trakhtenbrot}\label{sec:rice}

This section is context: three classical results, stated in the form in which the dynamic sections use them, with proofs in Appendix~A. Nothing in it is new; what is new begins in Section~\ref{sec:finite}.

\subsection{The semantic barrier: Rice}

We first state impossibility over open (infinite) domains.

\begin{theorem}[Rice, 1953]\label{thm2}
Let $P$ be a non-trivial property of partial computable functions. Then $I_P$ is undecidable.
\end{theorem}

\noindent\textit{Proof:} a standard many-one reduction from the Halting Problem (see Appendix~A.1).

\begin{corollary}[Semantic inviability of alignment over open domains]\label{cor1}
Evaluating any non-trivial safety property of a Turing-complete system is undecidable. No effective verifier can certify safe behaviour over the full input domain.
\end{corollary}

\textit{Note on scope. Rice's Theorem concerns semantic properties, determined solely by the input/output behaviour of $M_x$, independently of its internal description. Purely syntactic properties (`the description of $M_x$ has at most 100 states') are decidable and fall outside the theorem. It is tempting to conclude that alignment is entirely on the semantic side --- safety is about what the system does. Section~\ref{sec:preservation} shows that this is exactly where the static picture breaks down: the property `$M_x$ is safe and remains safe after self-modification' depends on the code of $M_x$ (two programs computing the same function can be rewritten into programs computing different functions), so it is neither semantic nor syntactic, Rice's Theorem does not apply to it, and it is undecidable nonetheless (Theorem~\ref{thm:int}). The scope of Rice's Theorem is not the scope of the obstruction.}

\subsection{The expressivity barrier: from completeness to G\"odel's incompleteness}\label{sec:godel}

The undecidability established by Rice concerns the behaviour of a running machine over all inputs. One might try to sidestep this with a static approach: fix a formal framework $F$ prior to execution and prove $F \vdash \varphi_{\mathrm{safe}}$. This section shows that the choice of $F$ itself leads to an inescapable dilemma between descriptive adequacy and deductive completeness --- a collapse whose precise form is G\"odel's Incompleteness.

A framework $F$ intended to reason about a general self-modifying system must, at minimum, interpret arithmetic. The designer faces a strict dichotomy. Under First-Order Logic ($\FO$), G\"odel's Completeness Theorem (1929) guarantees a sound and complete proof calculus, but by the L\"owenheim--Skolem Theorem no $\FO$ theory --- $PA$ itself included --- characterises $\N$ up to isomorphism: every $\FO$ theory with an infinite model has non-standard models, so first-order reasoning cannot pin down the intended structure a self-modifying system computes over. Under Second-Order Logic ($\SO$), quantification over sets and relations permits a categorical axiomatisation of $\N$, closing the gap left by $\FO$ --- but $\SO$ has no sound and complete proof calculus, a direct consequence of G\"odel's Incompleteness once $\SO$ can express arithmetic. Gaining descriptive power entails losing syntactic completeness. Any $F$ expressive enough to model a Turing-complete system must interpret at least Robinson's arithmetic $Q$ --- enough to represent all computable functions --- placing $F$ squarely within the scope of G\"odel's first theorem.

\begin{theorem}[G\"odel Incompleteness, 1931]\label{thm3}
Let $F$ be a consistent, effectively generated formal system that interprets Robinson's arithmetic $Q$. Then $F$ is incomplete: there exist sentences $\varphi$ such that $F \nvdash \varphi$ and $F \nvdash \neg\varphi$. (The second incompleteness theorem --- that $F$ does not prove its own consistency --- is not used in this article.)
\end{theorem}

\noindent\textit{Proof:} the standard Diagonalisation Lemma argument (see Appendix~A.2).

Consequence for static alignment verification: any framework $F$ adequate for specifying and reasoning about a Turing-complete system interprets $Q$ and is therefore subject to G\"odel's theorem. A static verifier operating in $F$ will therefore either be incomplete (missing genuine safety guarantees) or inconsistent (producing false certificates). The companion paper gives this a form that does not pass through any particular sentence of $F$ (Gumbau Mezquita, 2026b, \S8.3): for a safety property in the strict sense (`nothing bad ever happens', index set in $\Pi^0_1$ and non-trivial), the set of safe programs is co-r.e.\ and not recursive (Rice), hence not recursively enumerable. Hence no \emph{sound} proof system --- no recursively enumerable scheme of correct certificates --- covers the safe systems: a system may be safe without any certificate of it existing, although particular safe systems may well have proofs. This has the shape of G\"odel's unprovable truths, at the scale of the family of systems rather than of one system; we draw no conclusion from it about what any fixed $F$ proves of itself, and we do not use the second incompleteness theorem anywhere in this article. As Hern\'andez-Espinosa et al.\ (2026) note in a related context, the system resides in a formal zone where the framework cannot determine the safe decision.

\subsection{Trakhtenbrot's wall: the collapse of the finite strategy}\label{sec:trakh}

Sections~\ref{sec:rice} and~\ref{sec:godel} establish impossibility over infinite domains. A natural response restricts attention to the physical world, which is finite: perhaps safety can be certified by exhaustively checking all finite hardware configurations. We call this strategy \emph{Universal Finite Verification}: the paradigm that a system's safety can be guaranteed by confining its operation to a bounded, finite-state environment where all possible behaviours are, in principle, verifiable. This is essentially the position taken by Melo et al.\ (2025), who argue that restricting architectures to finite, halting configurations resolves the undecidability of alignment identified by Rice's Theorem --- since a machine guaranteed to halt over a finite domain can, in principle, be checked exhaustively, though, as they note, intractably. This section shows that Universal Finite Verification, too, is computationally infeasible --- not at the level of any particular finite structure, but at the level of universal quantification over all finite structures. As Corollary~\ref{cor2} below makes precise, the finite-and-halting restriction buys only a local, single-structure decidability; it does not extend to the universal guarantee a self-modifying system actually needs, since that guarantee must hold across whichever finite configuration the system's self-modification (Section~\ref{sec:preservation}) may produce next.

Let $L$ be a relational first-order vocabulary containing at least one binary relation symbol, and $\Fin(L)$ the class of all finite $L$-structures. A universal finite safety certificate for a sentence $\psi \in \FO(L)$ would establish $\modelsfin \psi \iff \forall M \in \Fin(L),\ M \models \psi$. The vocabulary used encodes a Turing machine's execution: $S(x,y)$ a successor relation, $T_q(t)$ that the machine is in state $q$ at time $t$, $H(p,t)$ that the head is at position $p$ at time $t$, $C_a(p,t)$ that cell $p$ contains symbol $a$ at time $t$.

\begin{theorem}[Trakhtenbrot, 1950]\label{thm4}
Let $L$ contain at least one binary relation symbol. The set of finitely valid sentences $V_{\mathrm{fin}} = \{\psi \in \FO(L) \mid \modelsfin \psi\}$ is not recursively enumerable.
\end{theorem}

\noindent\textit{Proof:} a many-one reduction from the Halting Problem to finite satisfiability (see Appendix~A.3).

\begin{corollary}[Finite universality is incomputable]\label{cor2}
Certifying safety for any specific, fixed finite structure is decidable, by exhaustive search within that structure. But certifying that a safety property holds across all possible finite hardware configurations is $\coRE$-complete, hence undecidable. This closes the finite-world route around Rice and G\"odel at the universal level.
\end{corollary}

\begin{remarknum}[on the halting-machine restriction]
Restricting to machines guaranteed to halt over a finite domain (Melo et al., 2025) resolves verification at the level it claims --- a fixed structure, decidable by exhaustive search --- and not beyond it: Corollary~\ref{cor2} relocates the difficulty to the universal level, and Section~\ref{sec:finite} shows that the local level is intractable in the worst case once the environment is succinctly described.
\end{remarknum}
\subsection{Summary of the static limits}\label{sec:staticconc}

Any algorithmic verification strategy for a safety property $P$ of a system $M_x$ evaluates it over some domain, infinite or finite. Over an unbounded domain it is blocked by Rice's theorem (no sound and complete verifier) and, if it relies on a formal framework interpreting arithmetic, by G\"odel's. Over a bounded domain, a certificate for one fixed finite structure is decidable by exhaustive search (Corollary~\ref{cor2}) and does not generalise, and a certificate quantifying over all finite configurations ---all configurations up to size $n$, for all $n$--- is the universal finite validity of Trakhtenbrot's theorem, not recursively enumerable. Probabilistic methods over infinite input spaces remain in the first case at the level of exact guarantees. This is a case analysis, not a theorem: its content is the three classical results, and the next section bounds the cost of the one case that remains decidable.

\section{Finite Structural Unverifiability and the Trilemma}\label{sec:finite}

This section stratifies the bounded-domain case into a tractability barrier, which is new, and the decidability barrier of Theorem~\ref{thm4}.

\begin{definition}[Robust specification over a finite arena]\label{def:arena}
A \emph{finite arena} is a finite directed graph of states, partitioned into \emph{system} states and \emph{environment} states, with an initial state and a set of terminal states each labelled \emph{safe} or \emph{unsafe}. A \emph{policy} $\sigma$ chooses, at each system state, one successor. A \emph{robust specification} is the requirement that, under $\sigma$, every play from the initial state --- whatever successors the environment chooses at environment states --- ends in a safe terminal. The arena is \emph{succinct} if its states are $\{0,1\}^{\le m}$ and its successor relation, partition, and labelling are computed by a circuit (or polynomial-time program) of size polynomial in $m$. Two problems: \textsc{Verify}: given a succinct arena and a policy $\sigma$ (as a circuit), does $\sigma$ satisfy the robust specification? \textsc{Realise}: given a succinct arena, does some $\sigma$ satisfy it? This is the standard setting of reactive synthesis (Pnueli \& Rosner, 1989). It is not the transition-operator setting of Section~\ref{sec:framework}: the environment's choices make the dynamics a relation, not a function $\Phi$, and the objects are finite arenas, not members of $\mathrm{MSM}$; the relation is discussed after Theorem~\ref{thm6}.
\end{definition}

\begin{theorem}[Finite Structural Unverifiability]\label{thm6}
No verification framework $V$ for robust safety over \emph{succinctly described} finite environments can simultaneously satisfy Soundness $(S)$, Completeness $(C)$, and Tractability $(T)$. This is structurally partitioned into two limits:
\begin{enumerate}[label=(\arabic*)]
\item \textbf{The Local Boundary (finite arenas, intractability):} on succinct arenas with plays of polynomial length, \textsc{Verify} is $\coNP$-complete and \textsc{Realise} is $\PSPACE$-complete. Hence, under $\PClass \neq \NP$ (resp.\ $\PClass\ne\PSPACE$), no polynomial-time procedure verifies a given policy (resp.\ synthesises a safe one) on this family.
\item \textbf{The Universal Boundary (Trakhtenbrot's Wall):} the universal finite verification problem, 
\[ \mathrm{UnivVerify}(M,\phi) \iff \forall A \in F,\ (M,A) \models \phi \] 
where $F$ is the class of all ordered finite structures and $\phi$ a first-order specification, is $\mathrm{coRE}$-complete (Theorem~\ref{thm4}). Guaranteeing universal safety over bounded domains is algorithmically uncomputable, unconditionally.
\end{enumerate}
\end{theorem}

\begin{proof}[Proof of Clause~1]
\emph{Hardness of \textsc{Realise}.} Given a $\Psi = Q_1 x_1 \ldots Q_n x_n.\ \varphi$ quantified Boolean formula, build in time $O(|\Psi|)$ the succinct arena $F_\Psi$ whose states are pairs (level $i$, partial assignment $a\in\{0,1\}^i$), system states at existential levels, environment states at universal levels, successors $(i{+}1,a\cdot 0)$ and $(i{+}1,a\cdot 1)$, terminals at level $n$ labelled safe iff $\varphi(a)$ holds. Let $W$ be the set of winning states: a terminal $(n,a)$ is winning iff $\varphi(a)$; a system state is winning iff some successor is; an environment state iff both are. By downward induction on $i$, $(i,a)\in W$ iff $Q_{i+1}x_{i+1}\ldots Q_nx_n.\,\varphi(a,\cdot)$ is true (base $i=n$ by definition; the existential step is ``or'', the universal step ``and''). A winning policy picks a winning successor at each winning system state, so a policy satisfying the robust specification exists iff $(0,\varepsilon)\in W$ iff $\Psi$ is true; the construction is $O(|\Psi|)$, and $\mathrm{TQBF}\le_p\textsc{Realise}$. \emph{Membership.} A depth-first evaluation of the game tree uses space $O(n\cdot|\Psi|)$ (recursion depth $n$, one partial assignment per frame); more generally, on a succinct arena with plays of polynomial length, \textsc{Realise} is decided in polynomial space by evaluating the alternating game recursively. (Without the polynomial-play promise, \textsc{Verify} becomes $\PSPACE$-complete even with no adversary --- succinct deterministic reachability --- and \textsc{Realise} $\EXPTIME$-complete.) \emph{\textsc{Verify}.} With $\sigma$ given as a circuit, a violating play is a polynomial-length sequence of environment choices, checkable in polynomial time; so \textsc{Verify} is in $\coNP$. Hardness: instances $\Psi=\forall\vec x.\,\varphi$ have no system states, the only policy is empty, and \textsc{Verify} asks whether $\varphi$ is a tautology: $\mathrm{TAUT}\le_p\textsc{Verify}$. The hardness of \textsc{Verify} is thus entirely the adversary's, and is present already for the empty policy; the policy matters for \textsc{Realise}.
\end{proof}

\begin{remarknum}[where the hardness comes from]
The hardness is the \emph{succinctness} of the arena: an exponential state space with a polynomial description. On an \emph{explicit} arena both problems are polynomial --- the winning region of a reachability or safety game is an attractor, computable in linear time, and alternating reachability on explicit graphs is $\PClass$-complete --- so the adversary alone contributes nothing. What the adversary determines is the \emph{class} once the arena is succinct: with polynomial-length plays and no environment states, \textsc{Realise} is $\NP$-complete (it is \textsc{Sat}); with alternation, $\PSPACE$-complete (it is TQBF); with plays of unbounded length, $\EXPTIME$-complete (Stockmeyer \& Chandra, 1979). Reachability on succinctly represented graphs with unbounded paths is itself $\PSPACE$-complete (Papadimitriou \& Yannakakis, 1986). The engineering consequence: making the state space explicit is what buys tractability, and it buys it even against an adversary; removing the adversary while the space stays succinct makes verification of a given policy polynomial but leaves synthesis $\NP$-complete.
\end{remarknum}

\begin{remarknum}[a model against a forbidden set]
The local boundary applies verbatim to a deployed model. Let a model be a deterministic Boolean circuit $C$ mapping a conversation history to a reply; bound the conversation by $n$ turns of messages of length at most $m$, both given in unary; and let $\mathrm{Bad}$ be a polynomial-time decidable set of replies. The problem \textsc{Model-Safe}: \emph{does $C$ never produce a reply in $\mathrm{Bad}$, for every sequence of user messages within the bounds?} is $\coNP$-complete. It is in $\coNP$ because a counterexample is a short message sequence, checked by simulating $C$ and $\mathrm{Bad}$; it is $\coNP$-hard because a circuit that, in one turn, reads the message as an assignment of a formula $\varphi$ and replies in $\mathrm{Bad}$ exactly when the assignment satisfies $\varphi$ is safe if and only if $\varphi$ is unsatisfiable. Three consequences. An adaptive user adds nothing: for a deterministic model, choosing each message in the light of the previous replies is choosing one sequence, and a stochastic model is handled by treating the seed as one more input. And any finite battery of tests is a sample of an astronomically large set of conversations. The limits of the statement are the honest ones: it says that no general efficient method exists (if $\PClass\ne\NP$), not that a given model is hard to check; it needs $\mathrm{Bad}$ to be formal and decidable, which real harms are not; and for circuits with billions of gates even a polynomial method is out of reach.
\end{remarknum}

\begin{remarknum}[scope and method]
\emph{Scope is existential}, exactly as the instrumentation class $\mathcal{R}_{\mathrm{inst}}$ of the companion paper: the family contains $\coNP$- and $\PSPACE$-complete instances; not every system is hard to verify; this suffices for the Trilemma. \emph{Method}: the hardness comes from the classical reductions, not from the descriptive-complexity correspondences (next remark); all the difficulty is $|F_\Psi|$ against $|\Psi|$. \emph{Relation to the dynamic sections}: a succinct arena is finite and its dynamics relational, so Clause~1 is not a statement about $\mathrm{MSM}$ or about the operator $\Phi$; it is a statement about the specification that a self-modifying system deployed in an environment is asked to meet, and it applies to any system, general or not, that is verified against such a specification. The two faces of unverifiability are, in this precise sense, about different objects: the decidability face about the system's own code, the resource face about the specification and the environment.
\end{remarknum}

\begin{remarknum}[a logical reading; finiteness is not tractability]
Descriptive complexity (Fagin, 1974; Immerman, 1986) reads these classes as logics over ordered finite structures: $\LFP$ captures $\PClass$ (Immerman, 1986; Vardi, 1982), existential second-order logic captures $\NP$ (Fagin, 1974), and $\PFP$ captures $\PSPACE$ (Vardi, 1982; Abiteboul \& Vianu, 1989). The winning region above is a monotone induction, definable in $\LFP$ and polynomial in the size of the arena; the intractability is entirely the gap between the arena and its description. Chess on an $N\times N$ board is the familiar example: finite, decidable by backward induction, and intractable --- $\PSPACE$-complete with a polynomial bound on moves (Storer, 1983), $\EXPTIME$-complete without it (Fraenkel \& Lichtenstein, 1981) --- because the board is a succinct description of an exponential game tree. The correspondences are a reading of Theorem~\ref{thm6}, not its proof.
\end{remarknum}

\begin{corollary}[The Soundness--Completeness--Tractability Trilemma]\label{cor3}
Assume $\PClass\ne\NP$. For any verification architecture $V$ for robust safety over succinctly described finite environments, the simultaneous intersection of Soundness, Completeness, and Tractability is empty: $\forall V,\ \neg(S(V) \wedge C(V) \wedge T(V))$. For universal verification over all finite environments, already $S\wedge C$ is empty, unconditionally.
\end{corollary}

\begin{proof}
Assume some $V$ satisfies $S \wedge C \wedge T$ on the local problem. Sound and complete means $V$ decides \textsc{Verify}; tractable means it does so in polynomial time; by Clause~1 this puts a $\coNP$-complete problem in $\PClass$, contradicting $\PClass\ne\NP$. For the universal problem, sound and complete means $V$ decides $\mathrm{UnivVerify}$, which is $\coRE$-complete by Clause~2, hence undecidable: no such $V$ exists, tractable or not. The two clauses are different obstructions: Clause~2 kills soundness-plus-completeness outright; Clause~1 kills tractability on top of them.
\end{proof}

This trilemma sharpens, into a structural necessity that rests on no sampling assumption, the Pareto frontier of alignment certification characterised by Agarwal (2026) over open-ended domains: because the verification task is bounded by $\coNP$-hard (and, for synthesis, $\PSPACE$-hard) and $\coRE$-complete constraints even in the finite case, no framework can simultaneously satisfy all three properties, in either the open or the bounded setting.

The static limits concern a fixed system: a single $M_x$, evaluated once, against a single $P$. We now ask whether a verified property persists across the system's own evolution and, if not, exactly where and why it fails to.

\section{A Formal Framework for Self-Modifying Systems}\label{sec:framework}

\begin{definition}[Transition operator]\label{def2}
A transition operator is a total computable function $\Phi : \N \to \N$. Given an index $x$, $\Phi(x)$ is the index of the machine obtained from $M_x$ by one step of self-modification. We impose no restriction on $\Phi$ beyond computability and totality, to capture self-improvement in its most general, unrestricted form.
\end{definition}

\begin{definition}[Evolution trace]\label{def3}
Given an initial index $x_0$ and a transition operator $\Phi$, the evolution trace is the sequence $\langle x_0, x_1, x_2, \ldots \rangle$ defined by $x_{t+1} = \Phi(x_t)$.
\end{definition}

\begin{remarknum}[on what, precisely, evolves]
No single function $\varphi_x$ changes over time --- each $\varphi_{x_t}$ is, as always, a fixed mathematical object. What evolves is which function the system computes: the trace indexes a sequence of distinct (or repeated) functions, and the system refers to this entire generative process --- the sequence together with $\Phi$ --- not to any single $\varphi_{x_t}$ in isolation. The undecidability results that follow concern the relation between consecutive members of this sequence, as mediated by $\Phi$, not any change internal to a fixed function.
\end{remarknum}

\begin{definition}[Semantic stasis]\label{def4}
An evolution trace is static from $t_0$ onward if $\varphi_{x_t} = \varphi_{x_{t_0}}$ for all $t \geq t_0$. Stasis is one instance --- not the criterion --- of the restricted expressivity that Corollary~\ref{cor4} identifies.
\end{definition}

\begin{definition}[Semantically well-defined (extensional) operator]\label{def5}
$\Phi$ is semantically well-defined, or \emph{extensional}, if $\varphi_x = \varphi_y \implies \varphi_{\Phi(x)} = \varphi_{\Phi(y)}$: the result of a transition depends only on the function computed by the input machine, not on the particular index describing it --- the natural analogue, for transition operators, of the restriction to semantic properties in Rice's Theorem. Otherwise $\Phi$ is \emph{intensional}: it reads the code.
\end{definition}

\begin{remarknum}[which regime is the real one]
Definition~\ref{def2} imposed no restriction on $\Phi$, and for good reason: a system that inspects and rewrites its own source code is, by that very description, intensional. Two programs computing the same function have different code and will, in general, be rewritten differently. The extensional case covers self-modification that is a function of behaviour alone; the intensional case, supplied by the companion paper (Gumbau Mezquita, 2026b), is the one self-modification actually lives in, and it is not a corollary of Rice's Theorem. We give both regimes below.
\end{remarknum}

\begin{definition}[$P$-disruption relative to $\femp$, extensional regime]\label{def6}
Let $\femp$ be the nowhere-defined function and $\eemp$ an index with $\varphi_{\eemp} = \femp$, with $\femp \notin P$. An extensional $\Phi$ is $P$-disruptive relative to $\femp$ if: (a) $\varphi_{\Phi(\eemp)} \notin P$; (b) there exists a witness $q$ with $\varphi_q \in P$ and $\varphi_{\Phi(q)} \notin P$; and (c) there exists a witness $p$ with $\varphi_p \in P$ and $\varphi_{\Phi(p)} \in P$. Condition (c) is needed: without it the elevated property may be empty, hence decidable (the extensional non-vacuity example below without (c) has $\Lambda_\Phi(P)=\emptyset$).
\end{definition}

\begin{definition}[Uniform disruption, intensional regime]\label{def:unif}
Let $P$ be non-trivial and fix an index $t$ with $\varphi_t \in P$. A transition operator $\Phi$ is \emph{uniformly disruptive} with respect to $P$ if there is a total computable \emph{wrapper} $w : \N \to \N$ such that, for every $e$: (i) $\varphi_{w(e)} = \varphi_t$ (all wrappers compute the same function: the parameter $e$ is inert data in the code); and (ii) $\varphi_{\Phi(w(e))} \in P$ if and only if $M_e(e)\!\downarrow$ (positive disruption), or if and only if $M_e(e)\!\uparrow$ (negative disruption). Condition (i) makes the disruption genuinely intensional: the wrappers are indistinguishable by behaviour and are separated by $\Phi$ according to a fact --- the halting of $M_e(e)$ --- that lives in their code. No uniformly disruptive $\Phi$ is extensional.
\end{definition}

\noindent\emph{Positive and negative disruption.} Negative disruption is the safety case: the update preserves $P$ unless a hidden event occurs. Positive disruption is its dual: an update that suspends $P$ until a hidden event restores it, as when a safeguard is installed only once a certificate is found. Which of the two is available depends on $P$: for a $\Pi^0_1$ property (a safety property) only negative disruption is possible, and for a $\Sigma^0_1$ property only positive disruption, since the reduction from $K$ (resp.\ $\overline K$) would otherwise place $K$ in $\Pi^0_1$ (resp.\ $\overline K$ in $\Sigma^0_1$). The theorems below cover both.

\begin{definition}[Pointwise valid confinement]\label{def7}
Given a verifier $V$ and property $P$, $V$ is valid at $t$ for a trace $\langle x_0, x_1, \ldots \rangle$ if $V(M_{x_t}, P) = \mathrm{True} \iff \varphi_{x_t} \in P$. For any single, fixed $t$ the verdict is one bit and such a $V$ trivially exists; uniformly in $t$ it does not (Rice). Pointwise validity is not in question. What is in question is whether the preservation of validity across a transition can be certified algorithmically, in advance of the transition having occurred.
\end{definition}
\section{The Undecidability of Safety Preservation}\label{sec:preservation}

We write, for a property $P$ and an operator $\Phi$,
\[
\Lambda_{\Phi}(P) := \{x \in \N \mid \varphi_x \in P \wedge \varphi_{\Phi(x)} \in P\} \;=\; I_P\cap\Phi^{-1}(I_P),
\]
the \emph{elevated property}: safe now, and safe after the next step --- the property intersected with the property transported along $\Phi$ (the pullback lemma of Gumbau Mezquita, 2026b, \S8.1). We call this \emph{preservation}; \emph{persistence} is the whole trajectory, $\{x \mid \forall k\ \varphi_{\Phi^k(x)} \in P\}$ (Gumbau Mezquita, 2026b, Definitions~3.1--3.2). (The biconditional set $\{x\mid\varphi_x\in P\iff\varphi_{\Phi(x)}\in P\}$ is a different set, not decidable together with $\Lambda_\Phi(P)$ in general; we do not use it.)

\begin{remarknum}[one step is a step-$0$ question]
$\Phi$ acts at compile time, so $\Lambda_\Phi(P)$ is a static set of indices: ``will the system be safe after its next self-modification?'' is ``is the compiled program $\Phi(x)$ safe?'', a property of the next version that is intensional when read on $x$ (Gumbau Mezquita, 2026b, \S2.2). Everything said here about preservation applies equally to compilation, instrumentation, or translation. The only genuinely dynamic question is persistence, and that is the one that can climb a level (Corollary~\ref{cor:step}).
\end{remarknum}

\begin{theorem}[Undecidability of preservation, extensional regime]\label{thm7}
Let $P$ be a non-trivial property of partial computable functions with $\femp \notin P$, and let $\Phi$ be extensional and $P$-disruptive relative to $\femp$. Then $\Lambda_\Phi(P)$ is not decidable. (Condition (c) alone suffices for this; conditions (a) and (b) are used for the conditional question of Corollary~\ref{cor4}.)
\end{theorem}

\begin{proof}
Since $\Phi$ is extensional (Definition~\ref{def5}), it descends to a well-defined map $\Phis$ on functions themselves: for any partial computable $g$, set $\Phis(g) := \varphi_{\Phi(x)}$ for any index $x$ with $\varphi_x = g$ --- the choice of $x$ does not matter, precisely by Definition~\ref{def5}. Define $Q := \{ g \text{ partial computable} \mid g \in P \wedge \Phis(g) \in P \}$. $Q$ is a set of partial computable functions with $I_Q = \Lambda_\Phi(P)$, since $\Phis(\varphi_x) = \varphi_{\Phi(x)}$. $Q$ is non-trivial: $Q\ne\emptyset$ by (c), and $Q\subseteq P\subsetneq\mathcal{PC}$ since $P$ is non-trivial (condition (b) is not needed for this). By Theorem~\ref{thm2}, $I_Q$ is undecidable.
\end{proof}

\begin{theorem}[Undecidability of preservation, intensional regime]\label{thm:int}
Let $P$ be a non-trivial property of partial computable functions and let $\Phi$ be uniformly disruptive with respect to $P$ (Definition~\ref{def:unif}). Then the self-halting set $K$ (positive disruption) or its complement (negative disruption) reduces many-one to $\Lambda_\Phi(P)$ via the wrapper $w$; hence $\Lambda_\Phi(P)$ is undecidable. Moreover $\Lambda_\Phi(P)$ is not a semantic property: there are indices $x, y$ with $\varphi_x = \varphi_y$, $x \in \Lambda_\Phi(P)$ and $y \notin \Lambda_\Phi(P)$, so Rice's Theorem does not apply to it.
\end{theorem}

\begin{proof}
(Gumbau Mezquita, 2026b, Lemma~3.8 and Theorem~4.1.) By (i), $w(e) \in I_P$ for every $e$. By (ii), $\Phi(w(e)) \in I_P$ exactly when $e \in K$ (resp.\ $e \notin K$). Hence $w(e) \in \Lambda_\Phi(P) \iff e \in K$ (resp.\ $e \notin K$), and $w$ is total computable: this is a many-one reduction. For the last claim take $e \in K$ and $e' \notin K$: $\varphi_{w(e)} = \varphi_{w(e')} = \varphi_t$ by (i), and exactly one of $w(e)$, $w(e')$ lies in $\Lambda_\Phi(P)$.
\end{proof}
\begin{remarknum}[natural instances]
Uniform disruption is not a curiosity. The companion paper exhibits an effective class of such rewrites, the instrumentation class $\mathcal{R}_{\mathrm{inst}}$: rewriters that wrap a program with a monitor or a patch whose behaviour depends on a hidden computation. An update that can write arbitrary code into the system can embed such a computation, so a system that rewrites its own code in a general language is in this situation. An update that only adjusts the weights of a fixed architecture, with bounded precision and context, is a finite-state system, and only the resource face of Section~\ref{sec:finite} reaches it.
\end{remarknum}
\begin{remarknum}[one obstruction, two scales]
Both proofs use the s-m-n reduction from the halting set that proves Rice's Theorem, and differ only in where it lands. Extensionally, the elevated property is semantic and Rice does all the work; intensionally, it is not, Rice cannot see it, and the reduction has to be carried out inside a single class of programs that all compute $\varphi_t$. This is the precise content of ``Rice one level up'': one engine, Kleene's s-m-n theorem, at two scales, the function and the code. The recursion theorem is not used.
\end{remarknum}

\begin{remarknum}[non-vacuity, both regimes]
\emph{Extensional regime.} $P$ = `halts on input $0$' and the shift $\varphi_{\Phi(x)}(z)=\varphi_x(z+1)$: $\Phi(\eemp)$ is nowhere defined (a); an index $q$ with $\varphi_q(0)\!\downarrow$, $\varphi_q(1)\!\uparrow$ gives (b); any total function gives (c).
\emph{Intensional regime.} $P$ = `computes the identity'; $w(e)$ returns its input and carries $e$ as an inert comment; $\Phi$ reads the comment and synthesises the program that simulates $M_e(e)$ and, if it halts, returns its input. Then $\varphi_{\Phi(w(e))}$ is $\mathrm{id}$ if $e \in K$ and $\femp$ otherwise. $\Phi$ is a pure rewriter of text; the simulation lives in the program it outputs. The companion paper's class $\mathcal{R}_{\mathrm{inst}}$ collects all such operators (Gumbau Mezquita, 2026b, Theorem~3.12). The switch is to $\femp$, not to an arbitrary unsafe witness, because a program cannot retract answers it has already given: a switch from $\varphi_g$ to $\varphi_t$ is realisable iff $\varphi_g \subseteq \varphi_t$ (Gumbau Mezquita, 2026b, Lemma~7.5). A system can only change what it has not yet done.
\end{remarknum}

\begin{corollary}[Safety--Generality Dichotomy: certification requires restricted expressivity]\label{cor4}
Under the hypotheses of Theorem~\ref{thm7} with conditions (b) and (c), or of Theorem~\ref{thm:int}, there is no \emph{total} algorithm that, given an index $x$ with $\varphi_x \in P$, decides whether $\varphi_{\Phi(x)} \in P$ will hold after the next self-modification step. Consequently, if for some transition operator $\Phi$ this conditional question \emph{is} decidable by a total algorithm, then $\Phi$ is not $P$-disruptive with (b) and (c) (extensional regime) and not uniformly disruptive (intensional regime). Algorithmic certification of preservation is available only for transition operators of restricted expressivity in this precise sense.
\end{corollary}

\begin{proof}
Intensional regime: Theorem~\ref{thm:int}, whose reductions land inside $I_P$ via the wrapper family, so the conditional question inherits $K$ or $\overline K$. Extensional regime: let $Q=\{g\in P:\Phis(g)\in P\}\ni\varphi_p$ and $Q'=\{g\in P:\Phis(g)\notin P\}\ni\varphi_q$ (by (c) and (b)). The index sets $I_Q$ and $I_{Q'}$ are recursively inseparable: by s-m-n let $x_e$ compute $\varphi_p$ if $\varphi_e(e)=0$ and $\varphi_q$ if $\varphi_e(e)=1$ (and $\femp$ otherwise, which is outside $P$ by (a)); a total decider of the conditional question would separate $\{e:\varphi_e(e)=0\}$ from $\{e:\varphi_e(e)=1\}$, which are recursively inseparable. The second sentence is the contrapositive.
\end{proof}

\begin{remarknum}[what the dichotomy says]
The corollary concerns \emph{total} deciders; a partial procedure correct on the promise $\varphi_x\in P$ may exist (with $P$ = `total' and $\Phis(g)=g$ if $g(0)\!\downarrow=0$, $\femp$ otherwise, ``output $[\varphi_x(0)=0]$'' is correct on every total $\varphi_x$ and diverges elsewhere). The dichotomy is not a statement about stasis: certification does not require the trace to become eventually static, since a sanitiser that always outputs a fixed safe program, a constant operator, or a menu of transformations none of which is disruptive all keep changing the system and have a decidable conditional question. What they share is restricted expressivity in the exact sense of the corollary, not stasis. Two glosses are not the criterion: ``cannot embed a computation'' (the shift embeds nothing and is $P_0$-disruptive), and ``a decidable, even finite, menu'' (one disruptive member suffices). The exact dichotomy: \emph{certifiable preservation and unrestricted self-modification are mutually exclusive}. It does not say that self-modifying systems are unsafe; it says that no total algorithm certifies, step by step, that they remain safe.
\end{remarknum}

\begin{corollary}[One step versus forever]\label{cor:step}
For every total computable $\Phi$ and every $P$, $\Lambda_\Phi(P) \le_T I_P$: certifying that a system stays safe across one step is never harder than certifying that it is safe. By contrast, there exist total computable $\Phi$ --- including \emph{extensional} ones, such as $\varphi_{\Phi(x)}(z) = \varphi_x(z+1)$ --- for which the persistence set $\{x \mid \forall k\ \varphi_{\Phi^k(x)} \in P\}$ is $\Pi^0_2$-complete for $P$ = `halts on input $0$' ($\Sigma^0_1$): one level above the property. The jump is specific to $\Sigma$-type properties: for $P\in\Pi^0_1$ the persistence set is $\forall k\,\forall s\,R(\Phi^k(x),s)$, still $\Pi^0_1$.
\end{corollary}

\begin{proof}
(Gumbau Mezquita, 2026b, Proposition~3.13, Theorem~4.5 and Remark~4.6.) The upper bound: $\Phi^{-1}(I_P)$ reduces to $I_P$ by the many-one reduction $\Phi$. The lower bound: with $\varphi_{\Phi(x)}(z) = \varphi_x(z+1)$, the persistence set for `halts on $0$' is exactly the set of total functions' indices, which is $\Pi^0_2$-complete.
\end{proof}

\begin{remarknum}[the property is fixed; the certificate is not]
The property $P$ is held fixed by design: a criterion that redefined itself to match the system would certify everything. The certificate must be recomputed at each step, and that is what Theorems~\ref{thm7} and~\ref{thm:int} deny a general guarantee for. Corollary~\ref{cor:step} quantifies the gap: each recomputation costs no more than the property, but the chain can be one level higher and never composes into a certificate of persistence --- the failure the supervisory tower of Section~\ref{sec:supervision} exhibits one level up.
\end{remarknum}
\section{Supervision by Further Systems Does Not Lower the Problem}\label{sec:supervision}

A natural engineering response to the impossibility results above is to delegate verification not to a fixed algorithm but to another, more capable system: if no simple checker can certify a general system, perhaps a second system --- at least as capable as the first --- can supervise it, a third can supervise the second, and so on. This is precisely the programme of scalable oversight (Amodei et al., 2016), which seeks to let weaker or comparable systems oversee stronger ones, and which underlies concrete proposals such as AI safety via debate (Irving et al., 2018), iterated amplification (Christiano et al., 2018), and recursive reward modelling (Leike et al., 2018), together with the recursive and nested variants in which one overseer is itself checked by another. This section formalises that strategy and shows that it does not escape the obstruction; it reproduces it at every level. The reason is exactly the pattern already central to this article: a supervisor powerful enough to be useful is powerful enough to fall under the same theorems it was meant to enforce.

\begin{definition}[Supervisory tower]\label{def8}
A supervisory tower is a sequence of systems $\langle M_0, M_1, \ldots \rangle$ in which $M_0$ is the system to be controlled and each $M_{n+1}$ is a supervisor tasked with verifying a non-trivial safety property $P$ of $M_n$ --- a computable procedure that, given an index, returns `safe', `unsafe', or does not halt --- and may consult the verdicts of the levels below it. A tower yields a \emph{total correct certificate of preservation} for $M_0$ if the verdict of its top level is total and agrees with $\Lambda_{\Phi_0}(P)$.
\end{definition}

\begin{theorem}[Supervision does not lower the problem: no finite tower certifies]\label{thm8}
Let $\Lambda_{\Phi_0}(P)$ be undecidable (by Theorem~\ref{thm7} or Theorem~\ref{thm:int}). Then no finite supervisory tower $\langle M_0, \ldots, M_N \rangle$ yields a total correct certificate of preservation for $M_0$, whatever the architecture, capability, or adaptivity of its levels.
\end{theorem}

\begin{proof}
The verdict of level $N$ is the output of a computable procedure that consults other computable procedures; a composition of computable procedures is computable. A total correct verdict would be a total computable decision procedure for $\Lambda_{\Phi_0}(P)$, contradicting the hypothesis. No property of the architecture enters: the undecidability is a property of the set, not of who tries to decide it. (Gumbau Mezquita, 2026b, Theorem~6.2.)
\end{proof}

\begin{remarknum}[what the supervisor need not be]
Theorem~\ref{thm8} assumes nothing about the supervisors: not that they are Turing-complete, not that they modify themselves. In particular it does not rest on the tempting but unproved implication `adequate $\Rightarrow$ Turing-complete'. The two propositions below say what any total supervisor cannot do, and what a self-modifying one inherits; the theorem above needs neither.
\end{remarknum}

\begin{proposition}[Every total supervisor errs on an undecidable set]\label{prop:bounded}
Let $V$ be any total supervisor (one that halts on every index; a resource bound is a special case, but any total $V$ has a computable running time), and let $R = \{x \mid V(x) = \mathrm{safe}\}$. Then $R$ is decidable, and the error set $\mathrm{Err}(V) := R \,\triangle\, \Lambda_{\Phi_0}(P)$ is infinite and undecidable: $V$ is wrong about infinitely many systems, and the set of systems it is wrong about cannot be decided either.
\end{proposition}

\begin{proof}
$V$ is total with computable verdict, so $R$ is decidable. If $\mathrm{Err}(V)$ were decidable, $\Lambda_{\Phi_0}(P) = R \,\triangle\, \mathrm{Err}(V)$ would be decidable. An undecidable set is infinite. (Gumbau Mezquita, 2026b, Lemma~6.4 and Proposition~6.5.)
\end{proof}

\begin{remarknum}[where the error lies]
Under uniform disruption the statement is explicit: the reduction $K \le_m \Lambda_{\Phi_0}(P)$ (or $\overline K\le_m$, by polarity) via the wrapper family shows that $V$ must err on infinitely many wrappers $w(e)$, and that $\{e \mid V \text{ errs on } w(e)\}$ is undecidable. No enlargement of $V$'s budget --- more time, more memory, more levels --- changes this; it moves the errors, never eliminates them, and never makes them recognisable. This is the precise form of the ``finite horizon'' intuition: it is not that the supervisor sees \emph{little}; it is that what it fails to see is not enumerable.
\end{remarknum}

\begin{proposition}[A general supervisor inherits the problem]\label{prop:general}
Let $M_{n+1}$ be a supervisor that is itself a Turing-complete system with its own transition operator $\Phi_{n+1}$, uniformly disruptive with respect to some non-trivial property $P_{n+1}$ (for instance, `issues correct verdicts on a fixed family'). Then $\Lambda_{\Phi_{n+1}}(P_{n+1})$ is undecidable and $K$ or $\overline K$ reduces to it, by polarity: the preservation of $M_{n+1}$'s own correct operation is exactly as unverifiable as that of the system it supervises. (Generality alone does not force this: a general supervisor with $\Phi_{n+1}=\mathrm{id}$ has trivial preservation; the hypothesis is the disruptive transition.) The paradigmatic case is the supervision of consistency of a deductive system, where the preservation of consistency under supervision is $\Pi^0_1$-complete (Gumbau Mezquita, 2026b, Proposition~7.28).
\end{proposition}

\begin{proof}
Theorem~\ref{thm:int} applied at level $n+1$.
\end{proof}

\begin{corollary}[Who verifies the verifier: the regress is the same obstruction, not a new one]\label{cor5}
The classical regress `who verifies the verifier?' is, under Theorem~\ref{thm8}, not an infinite sequence of distinct difficulties but a single obstruction --- the undecidability of preservation, Kleene's s-m-n reduction carried one level up --- that no architecture removes. A bounded supervisor decides a decidable set and errs invisibly (Proposition~\ref{prop:bounded}); a general one reproduces the obstruction at its own level (Proposition~\ref{prop:general}); a tower of either kind decides nothing that its bottom level could not. Delegating control of a general system to further systems relocates the point at which the uncertified gap appears; it never closes it. This regress is the formal counterpart of the `verifier regress' that Yampolskiy (2017) identified informally as the organising problem of a programmatic verifier theory; Theorem~\ref{thm8} shows that no finite tower of overseers lowers it, for a reason that does not depend on what the verifiers are.
\end{corollary}

\begin{remarknum}[scalable oversight is caught between two horns, and there is a third option]
Theorem~\ref{thm8} and the two propositions bound scalable oversight from both sides. If the overseeing system issues total verdicts (a human overseer modelled as a procedure, a smaller model, any bounded checker) then it decides a decidable set and, by Proposition~\ref{prop:bounded}, errs on an undecidable set of systems --- some self-modification violates $P$ entirely within a blind spot the overseer cannot even locate. If instead the overseer is itself a self-modifying system with a disruptive transition, Proposition~\ref{prop:general} applies: its own preserved correctness is unverifiable and stacking further overseers does not lower the problem. Debate, amplification, and recursive reward modelling, read as supervisory architectures, choose points on this spectrum; we do not formalise them here. The companion paper adds the one thing that \emph{can} be had, and it is not a certificate: for a safety property, the monitor that runs the system and raises an alarm when it observes a violation semidecides the violation, and no computable procedure semidecides preservation for all safe systems (Gumbau Mezquita, 2026b, Corollary~7.21). An overseer that promises to certify safety promises what cannot be delivered; an overseer that promises to detect the first violation, after the fact, promises exactly what can.
\end{remarknum}

\section{Bounded Verification Schemes}\label{sec:bounded}

Theorems~\ref{thm7} and~\ref{thm:int} concern exact, algorithmic certification. In practice, safety claims about learning systems are rarely of this form; they are statistical, temporally constrained, and observationally partial. This section shows that each of these practical alternatives is vulnerable to a constructive failure mode under evolution. Two of them ---finite test coverage and restricted observation--- are instances of one general theorem (Theorem~\ref{thm9}), stated for the properties that no finite set of inputs can settle; the third, the uncertified window, is a matter of latency and independent of it.

\begin{definition}[Effectively pointwise property]\label{def:nfd}
A property $P$ is \emph{effectively pointwise} if $P=\{g : \forall w\ g(w)\in P_w\}$ for a family of pointwise conditions $P_w\subseteq\N\cup\{\uparrow\}$, each non-trivial, together with two total computable functions $s,v$ with $s(w)\in P_w$ and $v(w)\notin P_w$ for every $w$: a satisfying value and a violating value, uniformly computable in $w$. `Total', `computes the identity', `never outputs $\rho$' are effectively pointwise; `$g(0)=0$' is not (its conditions are trivial for $w\ge1$), and the results of this Part are false for it. An effectively pointwise property is not determined by any finite set of inputs, and single-point changes at any $w$ move behaviours in and out of $P$; the constructions below use exactly these two facts, through $s$ and $v$.\end{definition}

\subsection{Statistical certification under evolution}\label{sec:statistical}

Fix, at each stage $t$, a finite test set $T_t \subseteq \N$ and define the empirical certification rate
\[
\hat{c}(x,T) = \frac{|\{w \in T : \mathrm{behaviour}(M_x,w) \in P_w\}|}{|T|},
\]
the fraction of tests the system passes.

\begin{proposition}[Indetectability of domain drift]\label{prop1}
Let $T_t = \{0,1,\ldots,t\}$ be a monotonically growing, bounded test sequence, and let $P_w$, for $w \in \N$, be pointwise correctness conditions with computable satisfying and violating values (so that $P$ is effectively pointwise, Definition~\ref{def:nfd}). There exists a computable sequence of machines $\langle M_{x_0}, M_{x_1}, \ldots\rangle$ and a transition operator $\Phi$ with $x_{t+1} = \Phi(x_t)$ such that for every $t$: (i) $\hat{c}(x_t,T_t) = 1$; (ii) $\varphi_{x_t}$ violates $P_w$ at $w = t+1$; and (iii) no verification scheme depending only on the behaviour of $M_{x_t}$ on $T_t$ can distinguish this trace, at any stage, from a trace correct on all of $\N$ (a scheme that reads the code sees the violating default branch; the claim is about behavioural evidence).
\end{proposition}

\begin{proof}
For $t \in \N$, let $M_{x_t}$ compute $\varphi_{x_t}(w) = $ a fixed output satisfying $P_w$ if $w \leq t$, and a fixed output violating $P_w$ if $w > t$ --- a direct table-plus-default construction, requiring no diagonalisation. By s-m-n, $t \mapsto x_t$ is computable, and $\Phi$ may be taken as any total computable extension of $\Phi(x_t) = x_{t+1}$ along this trace; no extensionality is needed for this or for any result of Section~\ref{sec:bounded}.

(i) Every $w \in T_t$ satisfies $w \leq t$, so $\hat{c}(x_t,T_t)=1$. (ii) $w=t+1>t$ always, so the violation is present at every stage, immediately past the boundary of the current test set. (iii) A scheme inspecting only the behaviour of $M_{x_t}$ on $T_t$ sees only values satisfying $P_w$; the violation at $w=t+1 \notin T_t$ lies outside its view --- stage $t$ is indistinguishable from a stage where $\varphi_{x_t}$ satisfies $P_w$ for all $w \in \N$.
\end{proof}

The point is not that this adversarial trace is likely to arise by accident, but that perfect empirical certification, repeated at every stage of an evolving system, is logically compatible with systematic, undetected failure just beyond the tested region --- at every stage, not only asymptotically, since $T_t$ is always finite while $\N$ is not. Proposition~\ref{prop:bounded} sharpens this: a total scheme not only misses violations outside its window; the set of systems on which it is wrong is itself undecidable.

\begin{remarknum}[an instance: empirically validated self-improvement]
The Darwin G\"odel Machine (Zhang et al., 2025) accepts a rewrite of a coding agent when the new agent performs well on a fixed suite of benchmarks: at each stage the evidence is behaviour on a finite test set, exactly the $T_t$ of Proposition~\ref{prop1}. For any effectively pointwise property the designers wish the lineage to keep, the proposition says that success on the suite at every stage is logically compatible with a violation just outside it at every stage; Theorem~\ref{thm9} says the same for any faithful scheme that certifies on finite behavioural evidence. The safety measures the authors adopt --- sandboxing and human review --- are, in the terms of Section~\ref{sec:practical}, Boxing and hierarchical supervision: they bound consequences and relocate the gap, and neither certifies preservation. In the verification hierarchy of Chen, Wang \& Qu (2026), benchmark validation sits below formal verification; even the top of that hierarchy certifies only membership in a subclass (Gumbau Mezquita, 2026b, Proposition~8.5). This is not a criticism of the system, which does not claim certification; it is a statement of what its validation step can and cannot establish.
\end{remarknum}

\subsection{The uncertified window}\label{sec:window}

We now identify a different, more elementary source of risk: a temporal window, between successive stages, during which the system already acts under a version of itself that has not yet been certified --- not because the certifying test was incomplete, but because certification simply takes time.

\begin{definition}[Verification scheme with temporal cost]\label{def9}
A verification scheme is a pair $(V,\tau)$, where $V$ is a pointwise verifier (Definition~\ref{def7}) and $\tau : \N \to \N$ bounds the number of steps $V$ requires to certify $M_x$, as a function of a size parameter of $x$.
\end{definition}

\begin{definition}[Unstalled operation]\label{def10}
The system operates in unstalled mode under $\Phi$ if, once $x_{t+1} = \Phi(x_t)$ is computed, $M_{x_{t+1}}$ begins producing outputs before $V$ completes certification of $M_{x_{t+1}}$. This is the opposite of the stalled regime --- re-certification with stasis between steps (Section~\ref{sec:practical}) --- in which the system halts and waits for $V$'s verdict before resuming.
\end{definition}

\begin{proposition}[Inevitability of the uncertified window]\label{prop2}
Let $(V,\tau)$ satisfy $\tau(t) > 0$ for all sufficiently large $t$. (This needs no conjecture: a non-trivial verifier must at least read its input; $\tau$ here is the cost at stage $t$, i.e.\ on input $x_t$. Whether it grows without bound depends on the scheme; Theorem~\ref{thm6} concerns arenas, not index verifiers, and is not invoked.) If the system operates in unstalled mode, then for every such $t$ there exists a temporal window $W_t$, of length at least $\tau(t) > 0$, during which $M_{x_{t+1}}$ executes while $V$ has not yet returned a verdict on $M_{x_{t+1}}$.
\end{proposition}

\begin{proof}
Immediate: under unstalled operation, execution of $M_{x_{t+1}}$ begins at or before $V$ begins certifying it; since $V$ requires at least $\tau(t)$ steps to return a verdict, there is an interval of at least that length during which $M_{x_{t+1}}$ is already executing and $V$'s verdict is not yet available.
\end{proof}

\begin{corollary}[The window is exactly where Theorem~\ref{thm7} bites]\label{cor6}
Under Proposition~\ref{prop2}, if $\Phi$ additionally satisfies the hypotheses of Corollary~\ref{cor4}, no algorithm --- not even one with complete knowledge of $M_{x_t}$ --- could have determined, prior to $W_t$, whether $\varphi_{x_{t+1}}\in P$: the property of the version that acts throughout the window. (Whether a violation occurs \emph{within} the finite window is decidable by simulation; the uniform undecidability is that of the property, not of the window event.)
\end{corollary}

\begin{proof}
Under the hypotheses of Theorem~\ref{thm7} with conditions (b) and (c), or of Theorem~\ref{thm:int}, Corollary~\ref{cor4} says that no total algorithm decides, given $x$ with $\varphi_x\in P$, whether $\varphi_{\Phi(x)}\in P$; in particular none does so uniformly along the trace. (Undecidability of $\Lambda_\Phi(P)$ alone would not suffice: a sanitiser has $\Lambda_\Phi(P)=I_P$ undecidable and a trivially decidable conditional question.) Since $W_t$ is, by Proposition~\ref{prop2}, the interval during which $M_{x_{t+1}}$ already acts while $V$'s verdict on it is still pending, the system's behaviour throughout $W_t$ is governed by a property whose status was, prior to $W_t$, undecidable from $x_t$ by any algorithm whatsoever.
\end{proof}

Proposition~\ref{prop2} and Corollary~\ref{cor6} anticipate the universal result of Section~\ref{sec:evadability}: a non-zero verification cost, combined with a system that does not wait, suffices to guarantee a window in which Theorem~\ref{thm7}'s undecidability is not merely a property of the formalism but a feature of the system's actual operation.

\subsection{Bounded observation}\label{sec:observation}

Confining a supervisor to weak logics, short horizons, or physical boxes --- the very mechanism securing decidability in Section~\ref{sec:rice} --- leaves any behaviour outside that restricted vocabulary invisible to it. We now make this precise, showing it is not an incidental shortcoming of any particular scheme but a consequence of finite observation alone.

We first fix the general notion of verification scheme that this section and the next will analyse.

\begin{definition}[Bounded verification scheme]\label{def11}
A bounded verification scheme is a triple $(V, O, \tau)$ where $O : \N \to \Pfin(\N)$ is computable, returning for each stage $t$ a finite observation window $O(t) \subseteq \N$; $\tau : \N \to \N$ is computable; and $V(M_{x_t}, P)$ is computed solely as a function of $\{\mathrm{behaviour}(M_{x_t}, w) : w \in O(t)\}$, in time at most $\tau(t)$. This subsumes the growing test sequence of Proposition~\ref{prop1} ($O(t) = T_t$) and the fixed observation window of Definition~\ref{def13} below ($O(t) = O_V$ for all $t$) as special cases.
\end{definition}

\begin{definition}[Faithful scheme]\label{def12}
$(V,O,\tau)$ is faithful if, whenever $\mathrm{behaviour}(M_x,w)$ satisfies $P_w$ for every $w \in O(t)$, $V(M_x,P) = \mathrm{True}$ --- $V$ does not reject evidence that is, by its own lights, entirely correct. (Observations not completed within the time budget $\tau(t)$ count as unobserved; faithfulness quantifies over what was actually observed.) Any verifier intended to certify rather than to obstruct satisfies this condition.
\end{definition}

\begin{remarknum}[faithfulness and finite determination]
Faithfulness asks $V$ to certify on the strength of finite evidence. Whether that can be sound depends on the property: if $P$ is settled by the window --- the case that Definition~\ref{def:nfd} excludes --- a faithful bounded verifier can be correct and complete (`$g(0)=0$' with window $\{0\}$). The results below assume $P$ is effectively pointwise; under that hypothesis no finite window settles membership, and certifying on finite evidence is certifying what the evidence cannot settle. The companion paper's polarity result (Gumbau Mezquita, 2026b, Proposition~7.15) is a statement about operators with anchors, not about window verifiers, and is not used here.\end{remarknum}

\begin{definition}[Restricted observation]\label{def13}
A verifier $V$ has observation window $O_V \subseteq \N$ if $V(M_x,P)$ is computed solely as a function of $\{\mathrm{behaviour}(M_x,w) : w \in O_V\}$. $V$ is $L$-restricted if $O_V$ is a proper subset of $\N$, fixed in advance and independent of $x$.
\end{definition}

\begin{proposition}[Invisibility under restricted observation]\label{prop3}
Let $V$ be $L$-restricted with window $O_V \subsetneq \N$, and $P$ effectively pointwise (Definition~\ref{def:nfd}). There exists a semantically well-defined $\Phi$ and an index $x$ with $\varphi_x \in P$ such that: (i) $\varphi_{\Phi(x)} \notin P$; (ii) $\mathrm{behaviour}(M_x,w) = \mathrm{behaviour}(M_{\Phi(x)},w)$ for every $w \in O_V$; hence (iii) $V(M_x,P) = V(M_{\Phi(x)},P)$ --- $V$ certifies identically before and after, detecting no change.
\end{proposition}

\begin{proof}
Fix $w^* \notin O_V$. Define $\Phi$ so that $\Phi(y)$ is the index of the machine agreeing with $M_y$ on every $w \neq w^*$, but producing on $w^*$ the violating value $v(w^*)$ regardless of $M_y$'s original behaviour there --- a uniform, computable single-exception modification, hence total computable and semantically well-defined by construction.

Choose $x$ such that $M_x$ satisfies $P_{w^*}$ correctly (possible since $P$ is non-trivial) and $\varphi_x \in P$. $M_{\Phi(x)}$ agrees with $M_x$ on all $w \neq w^*$ and violates $P_{w^*}$ at $w^*$, giving $\varphi_{\Phi(x)} \notin P$ --- (i). Since $w^* \notin O_V$, every $w \in O_V$ satisfies $w \neq w^*$, so behaviour agrees on all of $O_V$ --- (ii). Since $V$ depends only on behaviour within $O_V$, (ii) gives (iii) immediately.
\end{proof}

\begin{remarknum}[the modification is an augmentation]
The $\Phi$ of Proposition~\ref{prop3} adds one forbidden value at one unobserved input. This is the \emph{single-exception} rewrite: the companion paper's irrevocability lemma (Gumbau Mezquita, 2026b, Lemma~7.5) concerns switches inside a running program on inputs not yet answered; here the rewrite is at compile time and may overwrite a defined value, which is why it is always realisable for a rewriter even where it would not be for a monitor. That the move can be placed outside any finite window is the whole content of invisibility.
\end{remarknum}

\begin{corollary}[Weakness of bounded observation, formalised]\label{cor7}
Two facts, of different kinds, bound what a computable verifier can establish. \emph{Decidability}: any total verifier $V$ decides a decidable set and therefore, for a property with undecidable index set, errs on an undecidable set of systems (Proposition~\ref{prop:bounded}). \emph{Observation}: any verifier that depends only on behaviour within a fixed finite window, applied to an effectively pointwise property, is invisibly evaded by a single-exception rewrite outside its window (Proposition~\ref{prop3}). The first bounds what any total procedure can decide; the second bounds what behavioural evidence can show. Neither is removed by enlarging the window, and the second is not removed by unbounded resources.
\end{corollary}

A scheme that is not faithful rejects some system whose observed behaviour is entirely correct, and so fails on a fixed system with no dynamics at all; faithfulness is the minimal property that makes a verifier informative, and the only class for which evadability is a non-trivial question. Section~\ref{sec:evadability} treats it.

\subsection{Universal evadability}\label{sec:evadability}

Propositions~\ref{prop1} (Section~\ref{sec:statistical}) and~\ref{prop3} (Section~\ref{sec:observation}) are existence results: some trace evades a particular verification scheme exhibited in each case. The notion of a bounded verification scheme (Definition~\ref{def11}) already unifies both as special cases; we now prove that the resulting statement holds universally, for every faithful scheme of that general form rather than for the specific instances exhibited above.

\begin{theorem}[Universal evadability of faithful bounded schemes]\label{thm9}
Let $(V,O,\tau)$ be any faithful bounded verification scheme depending only on behaviour within its window, and $P$ effectively pointwise (Definition~\ref{def:nfd}). There exists a total computable $\Phi$ and an evolution trace $\langle x_0, x_1, \ldots\rangle$ with $x_{t+1}=\Phi(x_t)$, such that for every $t$: $V(M_{x_t},P) = \mathrm{True}$, while $\varphi_{x_t} \notin P$.
\end{theorem}

\begin{proof}
Since $(V,O,\tau)$ is fixed in advance, $O$ is known and may be used explicitly. For each $t$, let $M_{x_t}$ compute the function $w\mapsto s(w)$ for $w\in O(t)$ and $w\mapsto v(w)$ for $w\notin O(t)$; it satisfies every pointwise condition on $O(t)$ and violates every one outside. Since $O$, $s$ and $v$ are computable, the map $t \mapsto x_t$ is computable by s-m-n (and can be taken injective with decidable range), and $\Phi$ is any total computable extension along this trace, as in Proposition~\ref{prop1}.

Certification: by construction, $\mathrm{behaviour}(M_{x_t},w)$ satisfies $P_w$ for every $w \in O(t)$; since the scheme is faithful, $V(M_{x_t},P) = \mathrm{True}$. Genuine violation: fix any $w^*_t \in \N \setminus O(t)$; by construction $\mathrm{behaviour}(M_{x_t}, w^*_t)$ violates $P_{w^*_t}$, so $\varphi_{x_t} \notin P$.
\end{proof}

This resolves, for the class of faithful schemes, the general question: every bounded verification scheme admits a trace it certifies at every stage while genuinely and persistently violating the property in question, since any computable verifier must inspect some finite window at each stage, and finiteness alone already suffices for evasion. Faithfulness also states precisely what trusting a verifier can and cannot mean: we can trust that a faithful $V$ will not call correct behaviour incorrect, but not, in general, that when it says correct the system is in fact correct for all time --- the latter is what Theorem~\ref{thm9} shows no faithful bounded scheme can guarantee across an evolving trace. This sharpens the distinction between the undecidability of verification, a fact about every possible verifier, and the insecurity of any particular system, a fact about that system's actual behaviour that no theorem can settle in advance.

\begin{remarknum}[what the companion paper adds]
Theorem~\ref{thm9} exhibits one evading trace per scheme, for effectively pointwise properties. The companion paper's Corollary~7.8 is a different statement, about the elevated property $\Lambda_E(P)$ of an operator with a safety-polarity anchor: there, the set of systems that preserve $P$ is not recursively enumerable, so no scheme of certificates covers it. The two results are consistent and complementary: the first is about windows, the second about preservation under an anchored rewrite. Neither says that a safe system can never be certified --- a proof system may certify infinitely many; what neither a window nor a certificate scheme can do is cover them all.\end{remarknum}

\begin{remarknum}[a diagonal argument, in the lineage of Cantor and Turing]
The construction is diagonal in the exact sense of Cantor's diagonal and Turing's halting argument. In each, one is given a resource that can only cover a finite or enumerable part of an unbounded domain --- the digits an enumeration assigns, the predictions a purported halting-decider makes, or here the finite window $O(t)$ a computable verifier inspects --- and one builds an object that agrees with everything the resource covers while differing at a point it does not reach. Cantor's real differs at the $n$-th digit; Turing's machine differs on its own index; the evolving system here differs at any $w^* \in \N \setminus O(t)$. The mechanism is one and the same in all three: no finite or enumerable cover can exhaust an unbounded domain, so a diagonal point always remains. A computable verifier's finiteness is precisely such a cover, and the system's behaviour always has a `beyond' it has not looked at --- which is the same fact once more, now in its sharpest and most elementary form. That impossibility results of this kind rest on a shared diagonal, fixed-point construction --- in the lineage of Cantor, G\"odel and Lawvere --- has been noted at the level of proof technique (Brcic and Yampolskiy, 2023); the companion paper sharpens this into a shared generative operation, since the dynamic and supervisory barriers are one s-m-n reduction from the halting set, carried out once inside a class of behaviourally identical programs and read off thereafter --- ``diagonalisation amortised'', in the companion paper's phrase --- while the practical barriers of this Part are diagonal in Cantor's direct sense: a finite cover of an infinite domain always leaves a point uncovered, and no reduction from $K$ is needed.
\end{remarknum}

\textit{The window of Section~\ref{sec:window} remains relevant independently of Theorem~\ref{thm9}: even where a scheme resists the specific trace constructed there --- by refusing to certify until $O(t)$ exhausts all behaviour relevant to $P$, at the cost of $\tau(t)$ growing without bound --- an uncertified window of non-zero length remains unavoidable whenever the system does not stall to await an ever-slower verdict.}
\section{Practical Implications}\label{sec:practical}

Translating the static and dynamic barriers above into software engineering practice, the contemporary state of the art captures this structural collapse as a definitive impossibility result rather than a temporary limitation. This is consistent with the framework proposed by Seshia et al.\ (2022) in their paradigm of Verified AI, which posits that the gap between high-level safety specifications and low-level system behaviour forces engineers to navigate a rigorous trade-off between the depth of the properties verified and the complexity of the environments considered --- precisely the trade-off that Section~\ref{sec:finite} above grounds in explicit reductions, and that the taxonomy below realises as concrete design choices.

\begin{remarknum}[three failure modes, three causes]
Three risks identified above are sometimes conflated but require different remedies. The \emph{logical} mode is the undecidability of preservation (Theorems~\ref{thm7} and~\ref{thm:int}), addressed only by restricting $\Phi$ or $P$. The \emph{temporal} mode (Proposition~\ref{prop2}) is latency, addressed only by stalling. The \emph{observational} mode (Proposition~\ref{prop3}, Corollary~\ref{cor7}) has two causes --- a total verifier decides a decidable set, and behavioural evidence on a finite window cannot settle an effectively pointwise property --- and no widening of the window removes the second. Removing one mode leaves the others in place.
\end{remarknum}

This renunciation materialises in a taxonomy of containment architectures, developed originally for the static case (confining a fixed system to a decidable fragment) but which we now read against the static and dynamic results together. The Soundness--Completeness--Tractability Trilemma has exactly three vertices, and there are exactly three pure strategies --- one sacrificing each vertex; two further strategies pay nothing to the Trilemma and, for that very reason, do not attack the verification obstruction but operate outside it. Two further items are not containment strategies at all but changes of promise, and we list them separately.

\subsection{Three strategies that sacrifice a vertex of the Trilemma}

\paragraph{Logical Isolation (Shielding).} A rigid monitor formatted within a decidable logical fragment surrounds the agent (Alshiekh et al., 2018); in the classification of enforcement mechanisms of Hamlen, Morrisett \& Schneider (2006), it is an execution monitor, and the policies such monitors can enforce are exactly the co-r.e.\ ones (Viswanathan, 2000). The agent loses agency over decisions inexpressible in that language. Dynamically, this is the analogue of restricting $\Phi$ to a fixed set of permitted transformations none of which is disruptive (a decidable, or even finite, menu is not enough by itself: one disruptive member suffices): $\Phi$ is then not uniformly disruptive (Definition~\ref{def:unif}), Theorem~\ref{thm:int} does not apply, and preservation may become decidable. This is the exact dichotomy of Corollary~\ref{cor4} in constructive form. One restriction must not be confused with it: typed rewriting whose reduction steps preserve the computed function (strong normalisation in proof theory) is semantic stasis, the trivial case; and primitive-recursive rewriting in general restricts nothing here --- the shift of Corollary~\ref{cor:step} is primitive recursive, its trajectory does not terminate, and the anchor to the halting set is constructible with primitive-recursive rewriters, so $\Pi^0_1$-completeness of safety survives (Gumbau Mezquita, 2026b, Remark~8.4). On the resource face, restricting self-modification to an explicit, enumerable menu makes the trajectory space explicit rather than succinct, which is the source of the resource cost (Remark after Theorem~\ref{thm6}); bounding only the depth of self-modification does not, since an arena of polynomial depth is already succinct. \emph{Sacrifice}: completeness --- legitimate self-modifications outside the menu are barred.

\paragraph{Inversion of Burden (Proof-Carrying Code).} The agent must attach a formal proof of safety before acting; the verifier remains fast, but the agent cannot execute legitimate actions whose proof exceeds the allowed complexity. Dynamically: mandatory re-certification with stasis between steps, addressing only the temporal mode (Proposition~\ref{prop2}), not the logical one (Theorem~\ref{thm:int}). The companion paper states what a carried proof can be (Gumbau Mezquita, 2026b, \S8.3): a proof of safety attached to an action certifies membership in the subclass the proof scheme covers, never the property in general, since no recursively enumerable scheme covers all safe systems; over programs, Hamlen, Morrisett \& Schneider (2006) likewise identify the policies certifiable by proof-carrying code with the recursively enumerable ones. Proof-carrying code works exactly to the extent that it restricts what may be carried. \emph{Sacrifice}: bounded latency and availability.

\paragraph{Temporal Myopia (Bounded Horizon Planning).} Optimisation is limited to $k$ future steps, making the decision tree decidable at the cost of strategic blindness to long-term consequences --- the temporal analogue of the uncovered input window of Theorem~\ref{thm9}. \emph{Sacrifice}: soundness beyond the horizon (the guarantee holds only within it).

\subsection{Two strategies that operate outside the Trilemma}

\paragraph{Physical Confinement (AI Boxing).} The agent is converted into a closed oracle isolated from actuators and networks. It sacrifices no vertex of the Trilemma: what it reduces is the causal expressivity of the \emph{environment} --- the reach with which a failure can affect the world --- not any property of the verification problem. It neither detects misalignment nor makes it verifiable; it bounds its consequences. On the resource face it does one thing: without environment states a given policy has a single play, so \textsc{Verify} becomes polynomial; \textsc{Realise} stays $\NP$-complete on a succinct arena, and only an explicit arena makes both tractable. \emph{Price}: external reach, outside the Trilemma.

\paragraph{Hierarchical Supervision.} A supervisor at least as capable as the agent audits it, and another audits that one. It sacrifices nothing in the agent --- which is exactly why it does not contain the obstruction but relocates it upward. By Theorem~\ref{thm8}, no tower decides an undecidable set, whatever its levels are; by Proposition~\ref{prop:bounded}, a bounded overseer errs on a set of systems it cannot locate; by Proposition~\ref{prop:general}, a general overseer inherits the same undecidability at its own level. A finite tower has an uncertified top; an infinite one may converge pointwise to the right verdict (by the limit lemma, whenever the property is $\Delta^0_2$) but with no computable modulus of convergence: at no stage does one know that the verdict has settled. \emph{Price}: none --- and hence no containment.

\begin{table}[htbp]
\centering\small
\begin{tabularx}{\textwidth}{@{} l >{\raggedright\arraybackslash}X l >{\raggedright\arraybackslash}X @{}}
\toprule
Strategy & Price (Trilemma vertex) & Acts on & Leaves open \\
\midrule
Shielding & generality of $\Phi$ (C: some actions barred) & logical mode; succinctness & self-improvement capacity \\
Proof-carrying code & latency (T) & temporal mode & logical mode (Thm.~\ref{thm:int}) \\
Bounded horizon & exactness (S: unsound beyond horizon) & outside-window & long-term evasion \\
Boxing & causal reach (outside the Trilemma) & consequences & the whole obstruction \\
Hierarchical supervision & none & relocates upward & everything (Thm.~\ref{thm8}) \\
\bottomrule
\end{tabularx}
\caption{Containment strategies against the two faces of unverifiability, with the Trilemma vertex (S, C, T) each pure strategy sacrifices.}
\end{table}

\begin{remarknum}[where tractability can be bought]
The source of the resource cost is the succinctness of the arena (Remark after Theorem~\ref{thm6}). Making the arena explicit --- a bounded, enumerable menu of transformations and states --- brings robust verification into $\PClass$, adversary or not, since the winning region of an explicit game is computable in linear time. What the adversary and the length of plays determine is the class while the arena stays succinct: $\NP$ for existential choice with polynomial plays, $\coNP$ for verifying a fixed policy, $\PSPACE$ for alternation with polynomial plays, $\EXPTIME$ for unbounded alternating plays (Stockmeyer \& Chandra, 1979). Shielding pays on the succinctness side; Boxing, by removing the adversary, makes verification of a given policy polynomial but leaves synthesis $\NP$-complete on a succinct arena; only an explicit arena makes both tractable.\end{remarknum}

\subsection{Two changes of promise that are not containment}

\paragraph{Detection is optimal.} The results above are negative. The companion paper supplies the one positive statement that survives them (Gumbau Mezquita, 2026b, Corollary~7.21): for a safety property, the violation is semidecidable and the preservation is not, so a monitor that runs the system and raises an alarm when it observes a prohibited event detects every violation that occurs, after it occurs; no computable procedure semidecides preservation for \emph{all} safe systems (a proof search may certify infinitely many, never the totality), and none anticipates the violation. Runtime-verification monitors, shields, and watchdogs realise this bound. The classical logic bomb is the same object seen from the attacker's side: a program that behaves correctly until a finite event and then violates the policy is an elevation operator of safety polarity, detectable when the event fires and not before. A supervisor that promises to certify that a self-modifying system \emph{will} remain safe promises what cannot be delivered; one that promises to detect the first violation \emph{when} it happens promises exactly what can, and should be designed, evaluated, and trusted on that basis alone.

\paragraph{Fix the reference.} Comparing an evolving system against another evolving system (differential testing between two model versions, shadow deployment, distillation checks) inherits all of the hardness above and keeps the normal form ``still safe = in agreement and the trigger has not fired'' only under a condition the designer must enforce: the discrepancy trigger has to be confined to systems that can still answer, since on a system that has gone silent the trigger may fire without any discrepancy ever appearing (Gumbau Mezquita, 2026b, \S7.5). Comparing the evolving system against a \emph{frozen, total} reference --- an executable specification, a pinned version, a golden model --- removes the condition: the case becomes a literal instance of the companion paper's class, the preservation property has its normal form unconditionally, its complexity is known, and the only possible verdict direction (detect discrepancy, never certify agreement) is explicit. Good evaluation practice already compares against something that does not move; the theory says what is gained. Zenil et al.\ (2026) reaches a related conclusion by a different route: a closed self-training loop has no non-circular criterion of improvement without effective external grounding, of which a frozen reference is one instance.

A scheme most resembling current industry practice combines re-certification with statistical, rather than exact, guarantees: certify at 99.9\% confidence, halt, allow one evolution step, re-certify, and so on. Proposition~\ref{prop1} shows this does not escape the underlying tension: it converts an algorithmic impossibility into a silent reliability risk, since perfect performance on an evolving but always-bounded test suite is compatible with undetected failure just outside it.

In every case, escaping either the static barriers of Section~\ref{sec:rice} or Theorems~\ref{thm7} and~\ref{thm:int} in practice requires sacrificing exactly one of: the generality of the transition operator $\Phi$ or the expressivity of the monitored fragment, bounded verification latency, or the exactness of the safety guarantee. One thing cannot be traded for anything: for a safety property in the strict sense, violation can be watched for, and no sound procedure certifies all safe systems. This is the same Soundness--Completeness--Tractability Trilemma of Corollary~\ref{cor3}, reappearing, in a different guise, the moment the system under verification is allowed to change.

\paragraph{Two layers: actions and content.} For deployed agents the strategies above act on two layers that are verified differently. What an agent \emph{does} passes through an interface ---tools, permissions, execution environment--- whose enforcement mechanism can be small, explicit and formally verified: properties of actions can be certified by design, independently of the model, which is shielding and boxing applied to the interface. What an agent \emph{says} has no such envelope: properties of open-ended content (correctness, deception) are behavioural properties over an open domain, pre-deployment evaluations are finite tests (Theorem~\ref{thm9}), and output classifiers are total supervisors that err on a set they cannot locate (Proposition~\ref{prop:bounded}). For content, the honest promise is detection, through evaluation, independent review and monitoring.
\section{Discussion and Open Directions}\label{sec:discussion}

Since the systems of Section~\ref{sec:intro} are members of $\mathrm{MSM}$ (the scope remark of Section~\ref{sec:msm}), the barriers above are not abstract warnings but constraints on what can be promised of a real system. The operational consequence is not ``nothing is safe'' but four precise statements. (1)~\emph{Preservation certification is not a deferred engineering goal but a structural impossibility}: there are self-modifications --- realistic ones, Section~\ref{sec:preservation} --- for which no improvement in hardware, algorithms, or test suites makes it attainable (Section~\ref{sec:preservation}). (2)~\emph{A frozen instant can be certified; a trajectory cannot.} A pinned version over a fixed structure is verifiable (Corollary~\ref{cor2}); what no algorithm can do is carry that certificate across a change the system makes to itself. (3)~\emph{For a safety property, violation is what can be established uniformly, and only after the fact}: no procedure semidecides preservation for all safe systems, so a certificate of safety is always a certificate of membership in some subclass; a frozen instant can be certified (point~2), a trajectory cannot. (4)~\emph{Every practical guarantee is a point chosen on the tension, not a resolution of it}: responsible deployment does not consist in finding the containment that works, but in choosing consciously which vertex of the Trilemma is sacrificed, on which side (agent or environment) the resource cost is paid, and in making both explicit.

The value of the result is one of localisation: it draws precisely where the total guarantee breaks, and hence where ambition must move from certainty to risk management, and from certification to detection. That a gap exists for every faithful verifier is established with mathematical certainty; whether a particular deployed system's evolution ever falls into it is an empirical question about that system, which no universal theorem can answer in advance.

\paragraph{Open directions.} Three remain, each a precise question rather than a vague area.

\paragraph{Characterising the hard boundary.} Clause~1 of Theorem~\ref{thm6} establishes that the class contains $\PSPACE$-complete instances (an existential result). It remains open which alignment specifications are intractable and which admit efficient verification: to characterise the boundary between the two, and to determine whether deciding on which side a given specification falls is itself a hard problem. This would turn the warning ``hard cases exist'' into a map of which cases are hard.

\paragraph{Partial mitigations with quantified cost.} The taxonomy of Section~\ref{sec:practical} identifies what each strategy sacrifices, but not how much. The remark after Theorem~\ref{thm6} suggests the natural programme: index the verification problem by whether the trajectory space is explicit or succinct, by the presence of alternation, and by the length of plays, and locate each containment architecture on the resulting grid, whose known corners are $\PClass$ (explicit arena, with or without adversary; and succinct without adversary for verifying a given policy), $\NP$ (succinct, no adversary, synthesis), $\coNP$ (succinct, adversary, given policy), $\PSPACE$ (succinct, polynomial plays, alternation) and $\EXPTIME$ (succinct, unbounded plays, alternation). This would convert the binary dichotomy (general versus certifiable) into a Pareto frontier of achievable trade-offs with complexity classes at its vertices; each cell needs its own reduction, and we claim none beyond those four.

\paragraph{The decidability grid.} The companion paper leaves open the analogous question on the logical face --- a classification of transition operators by the degree of the transported property, and the study of operators with several triggers (Gumbau Mezquita, 2026b, \S9). Unifying the two faces --- resource regime against fixed-point structure --- in a single grid is left for synthesis work.
\section{Conclusion}\label{sec:conclusion}

Statically, no verifier certifies the safety of a fixed, sufficiently expressive system infallibly, completely, and tractably: over unbounded domains by Rice and G\"odel, over all finite configurations by Trakhtenbrot, and over succinctly described finite environments by the intractability of robust verification ($\coNP$) and synthesis ($\PSPACE$) --- the Trilemma (Corollary~\ref{cor3}). Dynamically, preservation of a safety property across one step of self-modification is undecidable in both regimes: by Rice's Theorem when the step reads only behaviour, and by the same s-m-n reduction inside a single behavioural class when it reads code, which is what self-modification does (Theorems~\ref{thm7} and~\ref{thm:int}). One step is never harder than the property; persistence along an unbounded trajectory can climb one level, as it does for $\Sigma^0_1$ properties (Corollary~\ref{cor:step}). Certification of preservation by a total algorithm is available only to transformations of restricted expressivity (Corollary~\ref{cor4}), and the restriction has to be imposed by design: whether an arbitrary system falls in a certifiable class cannot in general be checked afterwards from its code. Delegation does not help: no tower of supervisors decides an undecidable set, and every total supervisor errs on a set of systems it cannot locate (Section~\ref{sec:supervision}). Practical schemes that certify on finite behavioural evidence are evaded by construction for effectively pointwise properties (Theorem~\ref{thm9}).

Two things survive. A monitor that runs the system and raises an alarm on the first observed violation semidecides violation, while no sound procedure certifies all safe systems: supervision should promise detection, not certification. And comparison against a frozen, total reference has the full theory unconditionally, where comparison between two evolving systems has it only under a design condition.

These results do not say that self-modifying systems are unsafe; they say where the algorithmic guarantee breaks, what survives past that point, and which promises are the honest ones. Certifiability across self-modification and unrestricted self-modification cannot be jointly maximised. Containment strategies choose where to stand on that tension; they do not remove it. Whether a particular deployed system ever falls into the gap is an empirical question about that system, which no universal theorem answers in advance. $\blacksquare$
\begin{acks}
The structuring, mathematical correctness review, formal-language drafting, and orthotypographic review of this document were conducted with the assistance of artificial intelligence models (Gemini, Google, 2026; Claude, Anthropic, 2026). The definitions, results, and all content decisions are the author's, who takes full responsibility for them.
\end{acks}

\appendix

\section{Proofs of the Classical Results of Section~\ref{sec:rice}}\label{app:classical}

For self-containedness, this appendix reproduces the standard proofs of Rice's Theorem and G\"odel's Incompleteness Theorem in the form used in Sections~\ref{sec:rice} and~\ref{sec:godel}. Both results are classical and included here for completeness rather than as original contributions of this article.

\subsection{Proof of Theorem~\ref{thm2} (Rice, 1953)}

\begin{proof}
Many-one reduction from the Halting Problem. Let $\HP = \{\langle x,w\rangle \mid M_x \text{ halts on input } w\}$ denote the Halting Problem, known to be undecidable (Turing, 1936). We reduce $\HP$ to $I_P$.

Since $P$ is non-trivial, fix $q \in P$ and note that the nowhere-defined function $\femp$ satisfies $\femp \notin P$ (if instead $\femp \in P$, one exchanges the roles of $q$ and $\femp$ below). Given any pair $\langle x,w\rangle$, construct the machine $M'$ that, on input $y$: (1) simulates $M_x$ on $w$; (2) if and when that simulation halts, computes and returns $q(y)$. By construction, $\varphi_{M'} = q$ if $M_x$ halts on $w$, and $\varphi_{M'} = \femp$ otherwise. Hence $M' \in I_P \iff \langle x,w\rangle \in \HP$. A decision procedure for $I_P$ would therefore decide $\HP$; since $\HP$ is undecidable, no such procedure exists.
\end{proof}

\subsection{Proof of Theorem~\ref{thm3} (G\"{o}del Incompleteness, 1931)}

\begin{proof}
(Diagonalisation Lemma.) Since $F$ interprets $Q$, it represents every computable relation, in particular its own provability relation: there exists a formula $\mathrm{Prov}_F(x)$ that holds, within $F$, iff $x$ is the G\"odel number of a provable sentence. By diagonalisation, there exists a sentence $G$ with $F \vdash G \leftrightarrow \neg\mathrm{Prov}_F(\langle G \rangle)$ --- $G$ asserts its own unprovability.

If $F \vdash G$, then $G$ has a proof, so $F \vdash \mathrm{Prov}_F(\langle G\rangle)$ (provability is representable: $F$ proves that it proves what it proves), while by the biconditional $F \vdash \neg\mathrm{Prov}_F(\langle G\rangle)$, contradicting consistency; hence $F \nvdash G$. If $F \vdash \neg G$, then $F \vdash \mathrm{Prov}_F(\langle G\rangle)$ by the biconditional, asserting $G$ has a proof --- which, under $\omega$-consistency yields a contradiction; hence $F \nvdash \neg G$. Neither $G$ nor $\neg G$ is provable in $F$: $F$ is incomplete. (For the statement under mere consistency, as in Theorem~\ref{thm3}, one uses Rosser's sentence in place of $G$; the argument is standard and we do not repeat it.)
\end{proof}

\subsection{Proof of Theorem~\ref{thm4} (Trakhtenbrot, 1950)}

(Many-one reduction from the Halting Problem.) Given $M = \langle Q, \Sigma, \Gamma, \delta, q_0, q_{\mathrm{halt}}\rangle$, construct $\varphi_M$ over $L = \{S, T_q, H, C_a\}$ as the conjunction of four axiomatic blocks: (i) order and domain --- $S$ is a strict discrete linear order with a minimum and no cycles, modelling a finite time/tape grid; (ii) initial configuration --- at $t_0$, head at $p_0$, state $q_0$, tape initialised with the input; (iii) transition rules --- for each $\delta(q,a) = (q', a', d)$, the axiom
\[
\forall t\, \forall p \Big[ T_q(t) \wedge H(p,t) \wedge C_a(p,t) \implies T_{q'}(S(t)) \wedge C_{a'}(p, S(t)) \wedge H(p+d, S(t)) \Big],
\]
with inertia axioms for unwritten cells; (iv) halting --- $\exists t\, T_{q_{\mathrm{halt}}}(t)$.

By construction, $\varphi_M$ has a finite model iff $M$ halts, giving a many-one reduction from $\HP$ to finite satisfiability: $M$ halts $\iff \varphi_M$ is finitely satisfiable. Finite satisfiability is thus $\RE$-hard. Finite validity is the complement of finite satisfiability of $\neg\psi$, so --- since finite satisfiability is $\RE$-complete --- its complement, $V_{\mathrm{fin}}$, is $\coRE$-complete, hence not recursively enumerable: no algorithm can even list all finitely valid sentences, let alone decide them.

\end{document}